\documentclass[letterpaper,11pt]{article}

\usepackage{fullpage}
\usepackage[T1]{fontenc}
\usepackage{ae,aecompl}
\usepackage[dvips]{graphicx}
\usepackage{amssymb}
\usepackage{amsmath}

\usepackage{subcaption}
\usepackage[round,authoryear]{natbib}
\usepackage{epsf}
\usepackage{color}
\usepackage{array}
\usepackage{authblk}
\usepackage{colortbl}
\definecolor{webgreen}{rgb}{0,0.4,0}
\definecolor{webbrown}{rgb}{0.6,0,0}
\definecolor{purple}{rgb}{0.5,0,0.25}
\definecolor{darkblue}{rgb}{0,0,0.7}
\definecolor{darkred}{rgb}{0.7,0,0}
\usepackage[dvips,pdfborder=false]{hyperref}
\hypersetup{colorlinks,citecolor=darkblue,filecolor=black,linkcolor=darkred,urlcolor=webgreen,pdfpagemode=None,
pdfstartview=FitH}

\newcommand{\ignore}[1]{}

\usepackage{etoolbox}
\ifundef{\abstract}{}{\patchcmd{\abstract}%
    {\quotation}{\quotation\noindent\ignorespaces}{}{}}

\newtheorem{lemma}{{\sc Lemma}}

\newtheorem{theorem}{{\sc Theorem}}
\newtheorem{definition}{{\sc Definition}}

\usepackage{cleveref}
\crefname{claim}{claim}{claims}
\crefname{fact}{fact}{facts}
\crefname{algorithm}{algorithm}{algorithms}
\crefname{observation}{observation}{observations}
\crefname{equation}{equation}{equations}

\newenvironment{proof}{\noindent {\em Proof\/}:\enspace}
{\hfill $\blacksquare{}$ \medskip \\}

\DeclareMathOperator*{\argmax}{\arg\!\max}

\usepackage{pgf,tikz}
\usepackage{varwidth}
\usetikzlibrary{shapes,arrows, trees}
\usepackage{verbatim}
\usepackage{enumerate}
\usepackage{amssymb}
\usepackage{mathrsfs}
\usepackage{algorithm}
\usepackage[noend]{algorithmic}

\usepackage{multirow,bigdelim}
\usepackage{arydshln}
\usepackage{resizegather}

\title{{\bf Efficiency and Budget Balance \\in General Quasi-linear Domains}~\thanks{
A preliminary version of this work has appeared in the conference on Web and Internet Economics (WINE), 2016. This work is funded by a Fulbright-Nehru postdoctoral fellowship and the National Science Foundation under grants 1320620, 1546752, and 1617590, and by the ARO under award W911NF-16-1-0061.}
}

\author[1]{Swaprava Nath}
\author[1]{Tuomas Sandholm}

\affil[1]{\small Carnegie Mellon University\\ Computer Science Department\\ \texttt{\{swapravn,sandholm\}@cs.cmu.edu}}

\date{}

\begin{document}
\maketitle

\begin{abstract}
 We study {\em efficiency} and {\em budget balance} for designing mechanisms in general quasi-linear domains.  \citet{green1979incentives} proved that one cannot generically achieve both. We consider strategyproof budget-balanced mechanisms that are approximately efficient. For deterministic mechanisms, we show that a strategyproof and budget-balanced mechanism must have a {\em sink} agent whose valuation function is ignored in selecting an alternative, and she is compensated with the payments made by the other agents. We assume the valuations of the agents come from a bounded open interval. This result strengthens Green and Laffont's impossibility result by showing that even in a restricted domain of valuations, there does not exist a mechanism that is strategyproof, budget balanced, and takes every agent's valuation into consideration---a corollary of which is that it cannot be efficient. Using this result, we find a tight lower bound on the inefficiencies of strategyproof, budget-balanced mechanisms in this domain. The bound shows that the inefficiency asymptotically disappears when the number of agents is large---a result close in spirit to \citet[Theorem 9.4]{green1979incentives}. However, our results provide worst-case bounds and the best possible rate of convergence.

 Next, we consider minimizing any convex combination of inefficiency and budget imbalance. We show that if the valuations are unrestricted, no deterministic mechanism can do asymptotically better than minimizing inefficiency alone.

 Finally, we investigate randomized mechanisms and provide improved lower bounds on expected inefficiency. We give a tight lower bound for an interesting class of strategyproof, budget-balanced, randomized mechanisms. We also use an optimization-based approach---in the spirit of {\em automated mechanism design}---to provide a lower bound on the minimum achievable inefficiency of any randomized mechanism.

 Experiments with real data from two applications show that the inefficiency for a simple randomized mechanism is 5--100 times smaller than the worst case. This relative difference increases with the number of agents.
\end{abstract}

\noindent {\sc JEL Codes: D82, D71, D74} \\

\noindent {\sc Keywords}: quasi-linear preferences; efficiency; budget balance; affine maximizer; Green-Laffont impossibility.

\section{Introduction}
\label{sec:intro}

Consider a group a friends deciding which movie to watch together. Furthermore, the movie can be watched in someone's home by renting it or at any of a number of movie theaters. Each of these choices incurs a cost. Since individual preferences are different and sometimes conflicting, the final choice may not make everybody maximally satisfied. This may cause some of the agents to misreport their preferences or drop out of the plan. To alleviate this problem, one can think of monetary transfers so friends who get their more-preferred choice pay more than friends that get their less-preferred choice. Desirable properties of such a choice and payment rule are that
(1) the total side payments (transfers among the friends) should sum to zero, so there is no surplus or deficit, and (2) the choice is efficient, that is, the movie that is selected maximizes the sum of all the friends' valuations. Since the valuations are private information of the friends, an efficient decision requires the valuations to be revealed truthfully. This simple example problem is representative of many joint decision-making problems that often involve monetary transfers. Consider, for example, a group of firms sharing time on a jointly-owned supercomputer, city dwellers deciding on the location and choice of a public project (e.g., stadium, subway, or library), mobile service providers dividing spectrum among themselves, or a student body deciding which musician or art performer to invite to entertain at their annual function. These problems all call for efficient joint decision making and involve---or could involve depending on the application---monetary transfers.


This is a ubiquitous problem in practice and a classic problem in the academic literature.
We study the standard model of this problem where the agents' utilities are {\em quasi-linear}: each agent's utility is her valuation for the selected alternative (e.g., the choice of movie) minus the money she has to pay.
The classic goal is to select an {\em efficient} alternative, that is, the one that maximizes the sum of the agents' valuations (also known as {\em social welfare}).

In the setting where valuations are private information, a mechanism needs to be designed that incentivizes the agents to reveal their valuations truthfully (by the revelation principle, there is no loss in objective from restricting attention to such direct-revelation mechanisms). We will study the problem of designing strategyproof mechanisms, that is, mechanisms where each agent is best off revealing the truth regardless of what other agents reveal.

Even though there are mechanisms that select efficient alternatives in a truthful manner (e.g., the Vickrey-Clarke-Groves (VCG) mechanism~\citep{Vickrey1961,Clarke1971,Groves1973}), the transfers by the individuals do not sum to zero (in public goods settings, the VCG mechanisms leads to too much money being collected from the agents). The execution of such a mechanism needs an external mediator who consumes the surplus (or may need to pay the deficit), to keep the mechanism truthful and efficient---a phenomenon known as `money burning' in literature. In our movie selection example, this implies that we need a third party who will collect the additional money paid by the individuals, which is highly impractical in many settings. This has attracted significant criticism of the VCG mechanism~\citep{ROTHKOPF07}.
Ideally, one would like to design strategyproof mechanisms that are efficient and {\em budget balanced}, that is, they do not have any surplus or deficit.
\citet{green1979incentives} proved a seminal impossibility for this setting: in the general quasi-linear domain, strategyproof, efficient mechanisms cannot be budget balanced.

In this paper, we primarily focus on the problem of minimizing inefficiency subject to budget balance in the general setting of quasi-linear utilities.
This is because, in the applications of interest to this paper (e.g., movie selection), budget balance is more critical than efficiency. However, we show that for a large set of agents, the per-agent inefficiency vanishes. We also show that for deterministic settings, optimizing the sum (or any convex combination) of efficiency and budget balance---which seems to be the most sensible objective---does not provide any asymptotic benefit over maximizing efficiency subject to budget balance. The main contributions of this paper are summarized in the following subsection.

\subsection{Contributions of this paper}
\label{sec:contrib}

In this paper, we assume that the agents' valuations are picked from a bounded open interval.
In \Cref{sec:deterministic}, we characterize the structure of truthful, budget balanced, {\em deterministic} mechanisms in this restricted domain, and show that any such mechanism must have a {\em sink} agent,\footnote{Mechanisms using this idea have been presented with different names in the literature. The original paper by \citet{green1979incentives} refers to this kind of agents as a {\em sample} of the population. Later \citet{gary2000polling} formalized the randomized version of this mechanism which is known as {\em polling} mechanism. \citet{faltings2004budget} refers to this as an {\em excluded coalition} (when there are multiple such agents) and \citet{moulin2009almost} mentions this as {\em residual claimants}. However, we use the term `sink' for brevity and convenience, and our paper considers a different setup and optimization objective.} whose reported valuation function does not impact the choice of alternative and she gets the payments made by the other agents (\Cref{thm:sink}).
This result strengthens the Green and Laffont impossibility by showing that even in a restricted domain of bounded valuations, there does not exist a mechanism that is strategyproof, budget balanced, and takes every agent's valuation into consideration---a corollary of which is that it cannot be efficient. With the help of this characterization, we find the optimal deterministic mechanism that minimizes the inefficiency. This provides a tight lower bound on the inefficiency of deterministic, strategyproof, budget-balanced mechanisms. By inefficiency of a mechanism in this paper, we mean the worst-case inefficiency over all valuation profiles. We provide a precise rate of decay ($\frac{1}{n}$) of the inefficiency with the increase in the number of agents (\Cref{thm:det-bound}). This implies that the inefficiency vanishes for large number of agents.

To contrast this mechanism with the class of mechanisms that minimize budget imbalance subject to efficiency, in \Cref{sec:joint} we consider the joint objective of {\em efficiency-budget spillover}, which is a convex combination of inefficiency and budget imbalance. We prove that if the valuations are unrestricted, no deterministic, strategyproof mechanism can reduce this spillover at a rate faster than $\frac{1}{n}$ (\Cref{thm:unimprovability}). In other words, in the deterministic setting, minimizing the joint objective does not give any asymptotic advantage over the solution of minimizing inefficiency with the constraint that the mechanism is budget balanced.

We investigate the advantages of randomized mechanisms in \Cref{sec:random}. We first consider the class of {\em generalized sink} mechanisms. These mechanisms have, for every possible valuation profile, a probability distribution over the agents that determines 
each agent's chance of becoming the sink. This class of mechanisms is budget balanced by design. We show examples where mechanisms from this class are not strategyproof (\Cref{mech:irrel-sink}), and then isolate an interesting subclass whose mechanisms are strategyproof, the \emph{modified irrelevant sink mechanisms} (\Cref{mech:mod-irrel-sink}). We show that no mechanism from this class can perform better than the deterministic mechanisms if the number of alternatives is greater than the number of agents (\Cref{thm:gen-sink}). Since inefficiency (weakly) increases with the number of alternatives (\Cref{thm:incr-inefficiency}), we consider the extreme case of two alternatives and compare the performances of different mechanisms. We show that a na\"ive uniform random sink mechanism and the modified irrelevant sink mechanism (\Cref{mech:mod-irrel-sink})  perform equally well (\Cref{thm:naive-random,thm:mod-irrel-sink}) and reduce the inefficiency by a constant factor of $2$ from that of the deterministic mechanisms. However, the optimal, strategyproof, budget-balanced, randomized mechanism performs better than these mechanisms. Since the structure of strategyproof randomized mechanisms for general quasi-linear utilities is unknown,\footnote{For randomized mechanisms, results involving special domains are known, e.g., facility location~\citep{thang2010randomized,procaccia2009approximate,feldman2011randomized}, auctions~\citep{dobzinski2006truthful}, kidney exchange~\citep{ashlagi2013mix}, and most of these mechanisms aim for specific objectives.} we take a computational optimization-based approach to find the best mechanism for the special case of two agents. This approach is known in the literature as {\em automated mechanism design} \citep{conitzer2002complexity}. For an overview, see \citet{sandholm2003automated}. We prove that for a discrete valuation space with $3$ levels, the optimal inefficiency is reduced by a factor of $7$ (\Cref{thm:random-opt-lower}) from that of deterministic mechanisms. However, when the number of levels increases---thereby making the lower bound tighter to the actual open-interval problem---the improvement factor reduces to less than $5$ (\Cref{fig:lower-bound}). This is a significant improvement over the class of randomized sink mechanisms, which only improve over the best deterministic mechanism by a factor of 2.

We present experiments using real data from two applications. They show that in practice the inefficiency is significantly smaller than the worst case bounds (\Cref{sec:experiments}).
We conclude the paper in \Cref{sec:concl} and present future research directions. For a cleaner presentation, we defer most of the proofs to the appendix.

\subsection{Relationship to the literature}
\label{sec:literature}

The Green-Laffont impossibility result motivated the research direction of designing efficient mechanisms that are minimally budget imbalanced. The approach is to redistribute the surplus money in a way that satisfies truthfulness of the mechanism \citep{bailey1997demand,cavallo2006optimal}. The {\em worst case optimal} and {\em optimal in expectation} guarantees have been given for this class of mechanisms in restricted settings~\citep{guo2008optimal,moulin2009almost,guo2009worst}. The performance of this class of {\em redistribution} mechanisms has been evaluated in interesting special domains such as allocating single or multiple (identical or heterogeneous) objects~\citep{gujar2011redistribution}.
Also, mechanisms have been developed and analyzed that are budget balanced (or no deficit) and
minimize the inefficiency in special settings~\citep{masso2015cost,guo2014better,Mishra2016}. Characterization of strategyproof budget-balanced mechanisms in the setting of cost-sharing is explored by \citet{moulin2001strategyproof} and its quantitative guarantees are presented by \citet{roughgarden2009quantifying}. 

If the distribution of the agents' valuations is known and we assume common knowledge among the agents over those priors, the strategyproofness requirement can be weakened to Bayesian incentive compatibility. In that weaker framework,  mechanisms can extract full expected efficiency and achieve budget balance~\citep{dAspremont1979,arrow1979property}.
But these mechanisms need the knowledge of the priors over the valuations.

The general quasi-linear setting is important since there are settings, e.g., public goods, where the agents can have arbitrary valuations over the alternatives and the impossibility of Green and Laffont still holds.
Therefore, in the general quasi-linear setting, for mechanisms without priors, it is an important open question to characterize the class of strategyproof budget-balanced mechanisms, to find such mechanisms that minimize inefficiency, and to find strategyproof mechanisms that minimize the sum (or other convex combination) of inefficiency and budget imbalance. This paper addresses this important question in the general quasi-linear setting, for both deterministic and randomized settings. Our approach is also prior-free---the strategyproofness guarantees consider the worst-case scenarios. We show that the answers are asymptotically positive: even in such a general setup, the Green-Laffont impossibility is not too restrictive when the number of agents is large, and our mechanisms seem to work well on real-world datasets.

\section{Model and definitions}

We denote the set of agents by $N = \{1,2,\ldots,n\}$ and the set of alternatives by $A = \{a_1, a_2, \ldots, a_m\}$. We assume that each agent's valuation is drawn from an open interval $(-\frac{M}{2}, \frac{M}{2}) \subset \mathbb{R}$, that is, the valuation of agent $i$ is a mapping $v_i : A \to (-\frac{M}{2}, \frac{M}{2}), \forall i \in N$ and is a private information. Denote the set of all such valuations of agent $i$ as $V_i$ and the set of valuation profiles by $V = \times_{i \in N} V_i$.

A {\em mechanism} is a tuple of two functions $\langle f, \mathbf{p} \rangle$, where $f$ is called the {\em social choice function} (SCF) that selects the {\em allocation} and $\mathbf{p} = (p_1,p_2, \ldots, p_n)$ is the vector of {\em payments}, $p_i : V \to \mathbb{R}, \forall i \in N$. The utility of agent $i$ for an alternative $a$ and valuation profile $v  \equiv (v_i,v_{-i})$ is given by the {\em quasi-linear} function: $v_i(a) - p_i(v_i,v_{-i})$. For {\em deterministic} mechanisms, $f : V \to A$ is a deterministic mapping, while for {\em randomized} mechanisms, the allocation function $f$ is a lottery over the alternatives, that is, $f : V \to \Delta A$. With a slight abuse of notation, we denote $v_i(f(v_i,v_{-i})) \equiv \mathbb{E}_{a \sim f(v_i,v_{-i})} v_i(a) = f(v_i,v_{-i}) \cdot v_i$ to be the expected valuation of agent $i$ for a randomized mechanism.
The following definitions are standard in the mechanism design literature.

\begin{definition}[Strategyproofness]
\label{def:strategyproofness}
 A mechanism $\langle f, \mathbf{p} \rangle$ is {\em strategyproof} if for all $v \equiv (v_i, v_{-i}) \in V$,
 $$v_i(f(v_i,v_{-i})) - p_i(v_i,v_{-i}) \geq v_i(f(v_i',v_{-i})) - p_i(v_i',v_{-i}), \quad \forall \ v_i' \in V_i, i \in N.$$
\end{definition}

\begin{definition}[Efficiency]
\label{def:efficiency}
 An allocation $f$ is {\em efficient} if it maximizes social welfare, that is,
 $f(v) \in \argmax_{a \in A} \sum_{i \in N} v_i(a), \ \forall v \in V$.
\end{definition}

\begin{definition}[Budget Balance]
\label{def:budget}
 A payment function $p_i: V \to \mathbb{R}, i \in N$ is budget balanced if
 $\sum_{i \in N} p_i(v) = 0, \ \forall v \in V$.
\end{definition}
In addition, in parts of
this paper we will consider mechanisms that are oblivious to the alternatives---a property known as {\em neutrality}. To define this, we consider a permutation $\pi : A \to A$ of the alternatives. Therefore, $\pi$ over a randomized mechanism and over a valuation profile will imply that the probability masses and the valuations of the agents are permuted over the alternatives according to $\pi$, respectively.\footnote{We have overloaded the notation of $\pi$ following the convention in social choice literature (see, e.g., \citet{myerson2013fundamentals}). The notation $\pi(v)$ denotes the valuation profile where the alternatives are permutated according to $\pi$.}
\begin{definition}[Neutrality]
\label{def:neutrality}
 A mechanism $\langle f, \mathbf{p} \rangle$ is {\em neutral} if for every permutation of the alternatives $\pi$ (where $\pi(v) \neq v$) we have
 $$\pi(f(v)) = f(\pi(v)) \quad \text{ and } \quad p_i(\pi(v)) = p_i(v), \quad \forall v \in V, \forall i \in N.$$
\end{definition}
Note that efficient social choice functions are neutral and the Green-Laffont result implicitly assumes this property.

The most important class of allocation functions in the context of deterministic mechanisms are {\em affine maximizers}, defined as follows.
\begin{definition}[Affine Maximizers]
\label{def:affine-max}
 An allocation function $f$ is an {\em affine maximizer} if there exist real numbers $w_i \geq 0, i \in N$, not all zeros, and a function $\kappa : A \to \mathbb{R}$ such that
 $f(v) \in \argmax_{a \in A} \left( \sum_{i \in N} w_i v_i(a) + \kappa(a) \right).$
\end{definition}

As we will explain in the body of this paper, we will focus on {\em neutral} affine maximizers~\citep{mishra2012roberts}, where the function $\kappa$ is zero.
\begin{align}
 \label{eq:neutral-AM}
 f(v) \in \argmax_{a \in A} \sum_{i \in N} w_i v_i(a) & & \text{\bf neutral affine maximizer}
\end{align}

The following property of the mechanism ensures that two different payment functions of an agent, say $i$, that implement the same social choice function differ from each other by a function that does not depend on the valuation of agent $i$.\footnote{This definition is a generalization of auction revenue equivalence and is commonly used in the social choice literature (see, e.g., \citet{heydenreich2009characterization}).}
\begin{definition}[Revenue Equivalence]
\label{def:revenue-eq}
 An allocation $f$ satisfies {\em revenue equivalence} if for any two payment rules $p$ and $p'$ that make $f$ strategyproof, there exist functions $h_i : V_{-i} \to \mathbb{R}$, such that
 $$p_i(v_i,v_{-i}) = p_i'(v_i,v_{-i}) + h_i(v_{-i}), \ \forall v_i \in V_i, \forall v_{-i} \in V_{-i}, \forall i \in N.$$
\end{definition}

The metrics of inefficiency we consider in this paper are defined as follows.
\begin{definition}[Sample Inefficiency]
\label{def:sample-inefficiency}
The {\em sample inefficiency} for a deterministic mechanism $\langle f, \mathbf{p} \rangle$ is:
 \begin{equation}
 \label{eq:sample-inefficiency}
  \text{$r_n^M(f) := \frac{1}{nM} \sup_{v \in V} \left[ \max_{a \in A} \sum_{i \in N} v_i(a) - \sum_{i \in N} v_i(f(v)) \right].$}
\end{equation}
The metric is adapted to {\em expected sample inefficiency} for randomized mechanisms:
\begin{equation}
 \label{eq:sample-inefficiency-randomized}
  \text{$r_n^M(f) := \frac{1}{nM} \sup_{v \in V} \left\{ \mathbb{E}_{f(v)} \left[ \max_{a \in A} \sum_{i \in N} v_i(a) - \sum_{i \in N} v_i(f(v)) \right] \right\}.$}
\end{equation}
\end{definition}
The majority of this paper is devoted to finding strategyproof and budget balanced mechanisms $\langle f, \mathbf{p} \rangle$ that minimize the sample inefficiency.

A different, but commonly used, metric of inefficiency in the literature is the worst-case ratio of the social welfare of the mechanism and the maximum social welfare: $\inf_{v \in V} \frac{\sum_{i \in N} v_i(f(v))}{\max_{a \in A} \sum_{i \in N} v_i(a)}$. A conclusion similar to what we prove in this paper: ``{\em inefficiency vanishes when $n \to \infty$}'', holds in that metric as well, but unlike our metric, that metric would require an additional assumption that the valuations are positive, which is not always the case in a quasi-linear domain.

We are now ready to start presenting our results.
We begin with deterministic mechanisms that are strategyproof and budget balanced.

\section{Deterministic, strategyproof, budget-balanced mechanisms}
\label{sec:deterministic}

Before presenting the main result of this section, we formally define a class of mechanisms we call {\em sink} mechanisms. A sink mechanism has one or more
{\em sink} agents, given by the set $S \subset N$, picked a priori, whose valuations are not used when computing the allocation (i.e., $f(v) = f(v_{-S})$) and the sink agents do not pay anything and together they receive the payments made by the other agents.
The advantage of a sink mechanism is that it is strategyproof if it is strategyproof for the agents other than the sink agents and the surplus is divided among the sink agents in some reasonable manner, and sink mechanisms are budget balanced by design.
An example of a sink mechanism is where $S = \{i_s\}$ (only one sink agent) and $f(v_{-i_s})$ chooses an alternative that would be efficient had agent $i_s$ not exist, that is,
$f(v_{-i_s}) = \argmax_{a \in A} \sum_{i \in N \setminus \{i_s\}} v_i(a)$.
The \citet{Clarke1971} payment rule can be applied here to make the mechanism strategyproof for the rest of the agents---that is, for agents other than  $i_s$, $p_i(v_{-i_s}) = \max_{a \in A} \sum_{j \in N \setminus \{i_s, i\}} v_j(a) - \sum_{j \in N \setminus \{i_s, i\}} v_j(f(v_{-i_s})), \ \forall i \in N \setminus \{i_s\}.$  Paying agent $i_s$ the `leftover' money (that is, $p_{i_s}(v_{-i_s}) = -\sum_{j \in N \setminus \{i_s\}} p_j(v_{-i_s})$) makes the mechanism budget balanced.
Our first result establishes that the existence of a sink agent is not only sufficient but also \emph{necessary} for deterministic mechanisms.
\begin{theorem}
 Any deterministic, strategyproof, budget-balanced, neutral mechanism $\langle f, \mathbf{p} \rangle$ in the domain $V$ has at least one sink agent.\footnote{Green and Laffont's impossibility result holds for efficient mechanisms, and all efficient mechanisms are neutral. However, the neutrality of an efficient rule is upto tie-breaking, and Green-Laffont applies no matter how the tie is broken. Similarly, our result also holds irrespective of how the tie is broken. Therefore, this theorem covers and generalizes that result since having at least one sink agent implies that the outcome cannot be efficient.}
 \label{thm:sink}
\end{theorem}
The proof involves two steps.
\begin{enumerate}
 \item We leverage the fact that a mechanism that satisfies the stated axioms must necessarily be a neutral 
      affine maximizer (\Cref{eq:neutral-AM}) and has a specific structure for payments. The characterization of the payment structure comes from revenue equivalence.
 \item The core of the proof then lies in showing that for such payment functions, it is impossible to have no sink agents (identified as agents that have zero weights, $w_i = 0$, in the affine maximizer). This is shown in a contrapositive manner---assuming that there is no sink agent, we construct valuation profiles that lead to a contradiction to budget balance.
\end{enumerate}
The complete proof is given in the appendix.

Our next goal is to find the mechanism in this class that gives the {\em lowest} sample inefficiency (\Cref{eq:sample-inefficiency}). In the proof of the next theorem (presented in the appendix) we show that this is achieved when there is exactly one sink and the neutral affine maximizer weights are equal for all agents other agents. This, in turn, yields the following lower bound on inefficiency. 
\begin{theorem}
 For every deterministic, strategyproof, budget-balanced, neutral mechanism $\langle f, \mathbf{p} \rangle$ over $V$, $r_n^M(f) \geq \frac{1}{n}$. This bound is tight.
\label{thm:det-bound}
\end{theorem}

\section{Jointly minimizing budget imbalance and inefficiency}
\label{sec:joint}

In the previous section, we considered strategyproof, budget-balanced mechanisms that are minimally inefficient. We achieved a sample inefficiency lower bound of $\frac{1}{n}$.
\emph{Could one do better by, instead of requiring budget balance and minimizing inefficiency, relaxing budget balance by allowing money burning, and then minimizing the inefficiency from the allocation plus the inefficiency caused by money burning (or required subsidy from outside the mechanism)?}
This would seem like the sensible objective for minimizing overall waste.

In this section, we consider the joint problem of minimizing inefficiency and budget imbalance in the setting of deterministic mechanisms. We consider a convex combination of these two quantities since in the quasi-linear domain both of them contribute additively in the agents' utilities and social welfare. 
We assume that the combination proportions are normative constants with which the planner associates importance to the two factors of inefficiency and budget balance. Therefore, the coefficients of the convex combination are independent of the number of agents.
We show that for unrestricted valuations, i.e., when the valuation bound $M$ is large, considering this joint problem does not yield a better than $1/n$ rate of decay of the {\em efficiency-budget spillover} defined as follows.
\begin{align}
 \text{$\rho_n(f, \mathbf{p}) := \lim_{M \to \infty}\frac{1}{nM} \sup_{v \in V} \left[ \lambda \cdot T_1^n(f,v) + (1 - \lambda) \cdot T_2^n(\mathbf{p},v) \right],$} \label{eq:sample-joint-inefficiency}
\end{align}
Where $T_1^n(f,v) = \left( \max_{a \in A} \sum_{i \in N} v_i(a) - \sum_{i \in N} v_i(f(v)) \right)$ and $T_2^n(\mathbf{p},v) = \left| \sum_{i \in N} p_i(v) \right|.$
For $\lambda = 1$, that is, when budget imbalance is not a concern, one can use the VCG mechanism to get $\rho_n(f, \mathbf{p}) = 0$. Similarly, for $\lambda = 0$, a sink mechanism will give $\rho_n(f, \mathbf{p}) = 0$. So, the interesting cases are when $\lambda \in (0,1)$, and for this we have a solution that decays as $1/n$. In this section, we will assume that $\lambda, 0 < \lambda < 1$ is exogenous. Our goal is to find a strategyproof and neutral mechanism $\langle f, \mathbf{p} \rangle$ that minimizes the objective $\rho_n$.
We have shown in \Cref{sec:deterministic} that $T_1^n(f,v)$ can at most be a constant when $T_2^n(\mathbf{p},v)$ is zero for every $v$. Hence, for any improvement in the efficiency-budget spillover metric, that is, for $\rho_n(f, \mathbf{p}) = o(r_n(f))$, it is necessary that the term $\sup_{v \in V} \left[ \lambda \cdot T_1^n(f,v) + (1 - \lambda) \cdot T_2^n(\mathbf{p},v) \right]$ be $o(1)$. Since both $T_1^n(f,v)$ and $T_2^n(\mathbf{p},v)$ are non-negative, it is necessary that the factor $T_2^n(\mathbf{p},v) = o(1)$ for every $v \in V$. Our next result shows that it is impossible to have $T_2^n(\mathbf{p},v) = o(1), \ \forall v \in V \Leftrightarrow \lim_{n \to \infty} \sup_{v \in V} T_2^n(\mathbf{p},v) = 0$. Hence, for deterministic mechanisms with unrestricted valuations, the bound on inefficiency \emph{with no budget imbalance} (presented in \Cref{sec:deterministic}) is {\em asymptotically} optimal for this joint optimization problem as well.\footnote{It is easy to see that the conclusions of \Cref{thm:sink,thm:det-bound} hold even under the assumption of large $M$.}
\begin{theorem}[Unimprovability]
 \label{thm:unimprovability}
 For every deterministic, strategyproof, and neutral mechanism $\langle f, \mathbf{p} \rangle$ over $V$ and for every $\lambda \in (0,1)$, $\rho_n(f, \mathbf{p}) = \Omega \left( \frac{1}{n} \right)$. This bound is tight. For $\lambda = 0$, a sink mechanism, and for $\lambda = 1$, the VCG mechanism, achieves zero spillover.
\end{theorem}

\section{Randomized, strategyproof, budget-balanced mechanisms}
\label{sec:random}

In \Cref{sec:deterministic}, we saw that the best sample inefficiency achieved by a deterministic budget balanced mechanism is $\frac{1}{n}$. In this section, we discuss how the inefficiency can be reduced by considering randomized mechanisms.

An intuitive approach is to consider a mechanism where each agent is picked as a sink with probability $\frac{1}{n}$.
\begin{definition}[Na\"ive Randomized Sink]
 \label{def:nrs}
 A {\em na\"ive randomized sink} (NRS) mechanism picks every agent as a sink w.p. $\frac{1}{n}$ and takes the efficient allocation without that agent. The payments of the non-sink agents are VCG payments without the sink. The surplus is transferred to the sink.
\end{definition}
Clearly, this mechanism is strategyproof, budget balanced, and neutral by design.
One can anticipate that this may not yield the best achievable inefficiency bound.
Unlike deterministic mechanisms, very little is known about the structure of randomized strategyproof mechanisms in the general quasi-linear setting. Furthermore, we consider mechanisms that are budget-balanced in addition. Hence, even though we can obtain an upper bound on the expected sample inefficiency ($r_n^M(f)$) by considering specific mechanisms like the NRS mechanism described above, the problem of providing a lower bound (i.e., no randomized mechanism can achieve a smaller $r_n^M(f)$ than a given number), seems elusive in the general quasi-linear setting.

Therefore, in the following two subsections, we consider two approaches, respectively. First, we show lower bounds in a special class of strategyproof, budget-balanced, randomized mechanisms. Second, we analytically provide a lower bound of the optimal, strategyproof, budget-balanced, randomized mechanism for two agents and two alternatives, using a discrete relaxation of the original problem (in the spirit of \emph{automated mechanism design}~\citep{conitzer2002complexity,sandholm2003automated}). This approach provides an approximation to the strategyproofness and optimality of the original problem. However, the problems of finding a mechanism that matches this lower bound and extending the lower bound to any number of agents and alternatives are left as future work.

\subsection{Generalized sink mechanisms}

In the first approach, we consider a broad class of randomized, budget-balanced mechanisms, which we coin {\em generalized sink mechanisms}. In this class, the probability of an agent $i$ to become a sink is dependent on the valuation profile $v \in V$, and we consider mechanisms with only {\em one} sink, i.e., if the probability vector returned by a generalized sink mechanism is $g(v)$, then w.p.\ $g_i(v)$, agent $i$ is treated as the {\em only} sink agent. (One can think of a more general class of sink mechanisms where multiple agents are treated as sink agents simultaneously. However, it is easy to see---by a similar argument to that in the context of deterministic mechanisms, just before \Cref{lemma:lub}---that using multiple sinks cannot decrease inefficiency.) Clearly, the na\"ive randomized sink mechanism belongs to this class.
Once agent $i$ is picked as a sink, the alternative chosen is the {\em efficient} one {\em without} agent $i$. All agents $j \neq i$ are charged a Clarke tax payment in the world without $i$, and the surplus amount of money is transferred to the sink agent $i$. \Cref{alg:generalized-sink} shows the steps of a generic mechanism in this class.

\begin{algorithm}[H]
\caption{Generalized Sink Mechanisms, ${\cal G}$}
\label{alg:generalized-sink}
 \begin{algorithmic}[1]
 \STATE {\bf Input}: a valuation profile $v \in V$
 \STATE A generic mechanism in this class is characterized by a probability distribution over the agents $N$ (which may depend on the valuation profile), $g : V \to \Delta N$
  \STATE The mechanism randomly picks one agent $i$ in $N$ with probability $g_i(v)$
  \STATE Treat agent $i$ as the sink
 \end{algorithmic}
\end{algorithm}
Clearly, not every mechanism in this class is strategyproof. The crucial aspect is how the probabilities of choosing the sink are decided. If the probability $g_i(v)$ depends on the valuation of agent $i$, that is, $v_i$, then there is a chance for agent $i$ to misreport $v_i$ to have higher (or lower) probability of being a sink (being a sink could be beneficial since she gets all the surplus). For example, the {\em irrelevant sink} mechanism given in \Cref{mech:irrel-sink} is {\em not} strategyproof.
\begin{algorithm}[H]
\caption{Irrelevant Sink Mechanism (not strategyproof)}
\label{mech:irrel-sink}
 \begin{algorithmic}[1]
 \STATE {\bf Input}: a valuation profile $v \in V$
  \FOR{agent $i$ in $N$}
    \STATE Define: $a^*(v_{-i}) = \argmax_{a \in A} \sum_{j \neq i} v_j(a)$
    \IF{$\sum_{j \neq i} v_j(a^*(v_{-i})) - \sum_{j \neq i} v_j(a) > M$ for all $a \in A \setminus \{a^*(v_{-i})\}$}
      \STATE Call $i$ an irrelevant agent
    \ENDIF
  \ENDFOR
  \IF{irrelevant agent is found}
    \STATE Arbitrarily pick one of them as a sink with probability $1$
  \ELSE
    \STATE Pick an agent $i$ with probability $\frac{1}{n}$ and treat as sink
  \ENDIF
 \end{algorithmic}
\end{algorithm}
The intuition of this mechanism is that if agent $i$'s maximum valuation sweep ($-M/2$ to $M/2$) cannot change the alternative, this {\em irrelevant} agent can be selected as a sink, which yields the efficient alternative. However, when there is no such irrelevant agent, the decision of choosing every agent equi-probably leads to a chance of manipulation. An agent whose true valuation report does not lead her to become a sink can misreport a valuation so that there is no irrelevant agent, thereby increasing her own probability of being selected as a sink. For example, consider two valuation profiles with three agents (numbered $1,2,3$) and three alternatives ($a,b,c$), and $M = 1$. The agents' valuations in the first profile are $v_1 =  (0.5,  0,  -0.5), v_2 = (-0.5,  0, 0.5), v_3 = (0,  -0.5,  0.5)$,
and in the second profile they are $v_1' = (0.5,  0,  -0.5), v_2' = (-0.5,  0,   0.5), v_3' = (-0.5,  0,  0.5)$. The mechanism returns agent $1$ as the irrelevant agent in the first profile and therefore picks alternative $c$ with probability 1. There is no irrelevant agent in the second profile and hence each agent is picked as a sink with uniform probability, leading to the probability vector $(2/3,0,1/3)$ for the alternatives $a,b,c$. But agent 3 strictly gains by moving from the first profile to the second.\footnote{One can also verify that the weak monotonicity condition, which is a necessary condition for strategyproofness, is violated for agent 3 between these two profiles.}

A small modification of the previous mechanism leads to a strategyproof generalized sink mechanism. This shows that the class of generalized sink mechanisms is indeed richer than the constant probability sink mechanisms. In the modified version, we pick a default sink with a certain probability, which will be the sink if there exists no irrelevant agent among the rest of the agents. The change here is that when an agent is picked as a default sink, her valuation has no effect in deciding the sink. See \Cref{mech:mod-irrel-sink}.
\begin{algorithm}[H]
\caption{Modified Irrelevant Sink Mechanism (strategyproof)}
\label{mech:mod-irrel-sink}
 \begin{algorithmic}[1]
 \STATE {\bf Input}: a valuation profile $v \in V$
  \STATE Pick agent $i$ as a {\em default sink} with probability $p_i$
   \FOR{agent $j$ in $N \setminus \{i\}$}
    \IF{irrelevant agent(s) found within $N \setminus \{i\}$}
      \STATE Arbitrarily pick one of them as a sink
      \STATE Irrelevant agent is found
    \ENDIF
   \ENDFOR
   \IF{no irrelevant agent is found within $N \setminus \{i\}$}
    \STATE Treat agent $i$ as sink
   \ENDIF
 \end{algorithmic}
\end{algorithm}
It is easy to verify that this mechanism is strategyproof.
Interestingly, no generalized sink mechanism can improve the expected sample inefficiency over deterministic mechanisms if there are more alternatives than agents ($m>n$).
\begin{theorem}[Generalized Sink for $m>n$]
If $m>n$, every generalized sink mechanism has expected sample inefficiency $\geq \frac{1}{n}$.
\label{thm:gen-sink}
\end{theorem}
The proof is critically dependent on $m>n$. However, we can hope for a smaller inefficiency if the number of alternatives is small. We state this intuition formally as follows.
\begin{theorem}[Increasing Inefficiency with $m$]
\label{thm:incr-inefficiency}
 For every mechanism $f$ and for a fixed number of agents $n$, the expected sample inefficiency is non-decreasing in $m$, i.e., $r_{n,m_1}^M(f) \geq r_{n,m_2}^M(f), \forall m_1 > m_2$.\footnote{We overload the notation for the expected sample inefficiency $r_n$ with $r_{n,m}$ to make the number of alternatives explicit.}
\end{theorem}

\Cref{thm:gen-sink,thm:incr-inefficiency} suggest that in order to minimize inefficiency, one must have a small number of alternatives. So from now on, we consider the extreme case with $m=2$, where we investigate the advantages of randomization.

For two alternatives, the following theorem shows that the na\"ive randomized sink (NRS) mechanism reduces the inefficiency by a factor of two.
\begin{theorem}[Na\"ive Randomized Sink]
 \label{thm:naive-random}
 For $m=2$, the expected sample inefficiency of the NRS mechanism is $\frac{1}{n^2} \left \lceil \frac{n}{2} \right \rceil \sim \frac{1}{2n}$.
\end{theorem}
Even though the modified irrelevant sink (MIS) mechanism (\Cref{mech:mod-irrel-sink}) is sophisticated in its use of the valuation profile, it is easy to check that even that mechanism yields the same inefficiency on the profile illustrated in the proof above. 
\begin{theorem}[Modified Irrelevant Sink]
 \label{thm:mod-irrel-sink}
 For $m=2$, the expected sample inefficiency of the MIS mechanism (\Cref{mech:mod-irrel-sink}) is at least $\frac{1}{n^2} \left \lceil \frac{n}{2} \right \rceil \sim \frac{1}{2n}$.
\end{theorem}
Furthermore, it turns out that NRS and MIS have the same inefficiency on every valuation profile. Both mechanisms choose a single agent as a sink. The default sink for MIS is chosen uniformly at random, identical to the choice of the sink for NRS. If there does not exist an irrelevant sink in the rest of the agents, the inefficiency remains the same as that for the default sink, which is identical to the inefficiency of NRS for that choice of sink. But even if an irrelevant sink exists, by the construction of the irrelevant sink, the resulting alternative is the efficient alternative for the agents except the default sink. This outcome would have resulted even if the default sink was chosen as the sink. Therefore, the inefficiencies in MIS and NRS mechanisms are the same.

The above result does not say much about the lowest achievable expected sample inefficiency (even in this special class of generalized sink mechanisms). 
In order to understand the limit of lowest achievable inefficiency for randomized mechanisms, we take an optimization-based approach. However, to keep the analysis simple and tractable, we focus on the special case of two agents and two alternatives.
Our next result gives a lower bound on the inefficiency for the class of generalized sink mechanisms in that setting.
Since we now fix the number of agents in the analysis, minimizing the expected sample inefficiency is equivalent to minimizing the expected {\em absolute} inefficiency given by $n r_n^M(f)$ which is $\frac{1}{M} \sup_{v \in V} \left\{ \mathbb{E}_{f(v)} \left[ \max_{a \in A} \sum_{i \in N} v_i(a) - \sum_{i \in N} v_i(f(v)) \right] \right\}$.
Without loss of generality we will assume $M=1$. For the rest of this section, we let `inefficiency' mean the expected absolute inefficiency.
\begin{theorem}[Lower Bound of Generalized Sink]
\label{thm:LB-gen-sink}
 For $n = m = 2$, the expected absolute inefficiency of every strategyproof generalized sink mechanism is lower bounded by $\frac{1}{2}$.
\end{theorem}

\subsection{Unrestricted randomized mechanisms}

We now move on to study optimal randomized mechanisms without restricting attention necessarily to generalized sink mechanisms.
%
Finding a mechanism that achieves the minimum absolute inefficiency can be posed as the following optimization problem.
\begin{equation}
\label{eq:optimization}
 \begin{aligned}
  &\underset{f,\mathbf{p}}{\min}
  && \sup_{v \in V} \left[ \max_{a \in A} \sum_{i \in N} v_i(a) - \sum_{i \in N} v_i(f(v)) \right] \\
  &\text{s.t.}
  &&v_i(f(v_i,v_{-i})) - p_i(v_i,v_{-i})
  \\
  &&& \qquad
  \geq v_i(f(v_i',v_{-i})) - p_i(v_i',v_{-i}), \ \forall v_i, v_i', v_{-i}, \forall i \in N \\
  & && \sum_{a \in A} f_a(v) = 1, \ \forall v \in V, \\
  & && \sum_{i \in N} p_i(v) = 0, \ \forall v \in V, \\
  & && f_a(v) \geq 0, \ \forall v \in V, a \in A.
 \end{aligned}
\end{equation}
The objective function denotes the absolute inefficiency. The first set of inequalities in the constraints denote the strategyproofness requirement, where the term $v_i(f(v)) = v_i \cdot f(v)$ denotes the expected valuation of agent $i$ due to the randomized mechanism $f$. The second and last set of inequalities ensure that the $f_a(v)$'s are valid probability distributions, and the third set of inequalities ensure that the budget is balanced. The optimization is over the social choice functions $f$ and the payments $\mathbf{p}$, where the $f$ variables are non-negative but the $p$ variables are unrestricted.
Clearly, this is a linear program (LP), which has an uncountable number of constraints (because the equalities and inequalities have to be satisfied at all $v \in V$, which are the profiles of valuation functions mapping alternatives to an open interval).
We address this optimization problem using finite constrained optimization techniques by discretizing the valuation levels. We assume that each agent's valuations are uniformly discretized with $k$ levels in $[-M/2,M/2]$, which makes the set of valuation profiles $V$ finite. The optimal value of such a discretized relaxation of the constraints provides a lower bound on the optimal value of the original problem. This is because the discretized relaxation of the valuations only increases the feasible set since some of the constraints are removed, that is, more $f$'s and $p$'s satisfy the constraints, allowing a potentially lower value to be achieved for the minimization objective.
We now prove a lower bound when the number of discretized levels is three. The analysis uses a primal-dual argument on the discrete relaxation of the problem.
\begin{theorem}[Lower Bound of Inefficiency for Randomized Mechanisms]
\label{thm:random-opt-lower}
 For $n = m = 2$, and for $k=3$ discrete levels of valuations, the absolute inefficiency is lower bounded by $\frac{1}{7} = 0.142857$.
\end{theorem}
%
%
\begin{figure}[h!]
 \centering
 \includegraphics[width=0.6\linewidth]{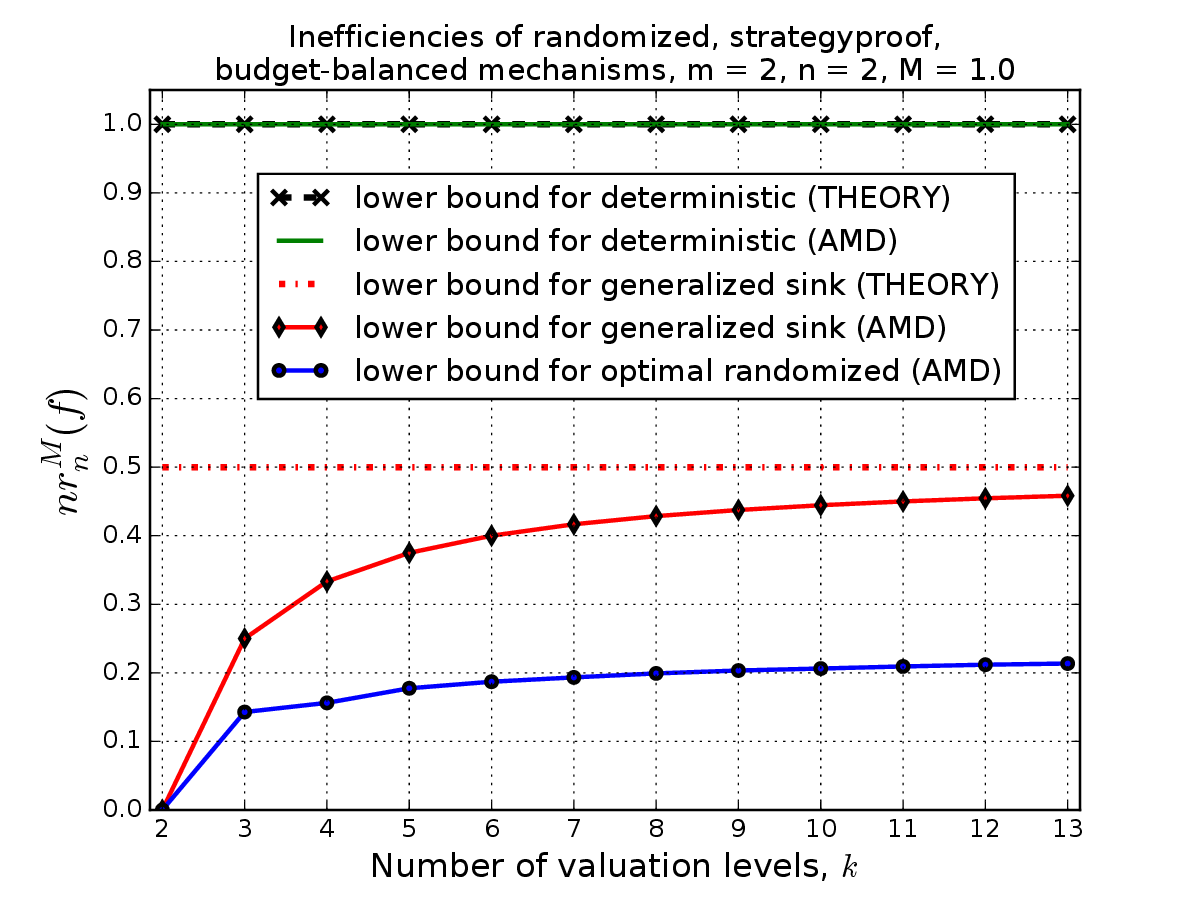}
\caption{Lower bound for the discrete relaxation of the inefficiency minimization LP.}
\label{fig:lower-bound}
\end{figure}
The proof technique can be extended to a larger number of discrete levels to obtain a tighter lower bound on the actual inefficiency. We conducted a form of automated mechanism design~\citep{conitzer2002complexity,sandholm2003automated} by solving this LP using \citet{gurobi} for increasing values of $k$. We apply the same optimization-based approach for generalized sink and the deterministic cases as well, even though for these cases we have theoretical bounds. The solid lines in \Cref{fig:lower-bound} show the optimization-based results (denoted as AMD) and the dotted lines show the theoretical bounds. Note that for deterministic case, the theoretical and optimization-based approaches overlap since the inefficiency is unity even with two valuation levels. The convergence of the optimization-based approach for generalized sink mechanism shows the efficacy of the approach and helps to predict the convergence point for the optimal randomized mechanism. One can see that the lower bound is greater than $0.2$ for the optimal mechanism, but it seems to converge to a value much lower than $0.5$.

\section{Experiments with real data}
\label{sec:experiments}

Even though the na\"ive randomized sink (NRS) (\Cref{def:nrs}) 
mechanism gives a worst-case sample inefficiency between $1/2n$ (for $2$ alternatives, \Cref{thm:naive-random}) and $1/n$ (for more alternatives than agents, $m > n$, \Cref{thm:gen-sink}), in this section we investigate its average and worst-case performances on real datasets of user preferences. Going back to the example of movie selection by a group of friends (\Cref{sec:intro}), we consider several sizes of the group. A small group consists of tens of friends, while if the decision involves screening a movie at a school auditorium, the group size could easily be in the hundreds. This is why we consider group sizes spanning from 10 to 210 in steps of 50. 

A similar situation occurs when a group of people decides which comedian/musician to invite in a social gathering, where they need to pay the cost of bringing the performer.

Keeping these motivating situations in mind, we used two datasets that closely represent the scenarios discussed. We used the MovieLens $20$M dataset~\citep{harper2016movielens} and the Jester dataset~\citep{goldberg2001eigentaste} to compare the performance of the two mechanisms with their worst case bounds. The first dataset contains preferences for movies, while the second contains preferences for online jokes. The MovieLens $20$M dataset ({\tt ml-20m}) describes users' ratings between 1 and 5 stars from MovieLens, a movie recommendation service. It contains 20,000,263 ratings across 27,278 movies. These data were created from the ratings of 138,493 users between January 09, 1995 and March 31, 2015. For our experiment, we sampled the preferences of a specific number of users (shown as agents on the x-axis of \Cref{fig:real-data}) multiple times uniformly at random from the whole set of users that rated a particular genre of movies, and computed the sample inefficiency on this sampled set and plotted the average expected sample inefficiency and standard deviation.

The Jester dataset ({\tt jester-data-1}) used in our experiment contains data from 24,983 users who have rated 36 or more jokes, a matrix with dimensions 24983 X 100, and is obtained from Jester, an online joke recommendation system.\footnote{In both datasets there are missing values because a user has typically not rated all movies/jokes. Before our experiment, we filled the missing values with a random realization of ratings drawn from the empirical distribution for that alternative (movie or joke).
The empirical distribution of an alternative is created from the histogram of the available ratings of the users.
We cleaned the dataset by keeping only those alternatives that have at least 10 or more available ratings and filled the rest using their empirical distributions.}

\Cref{fig:real-nrs} shows that the real preferences of users yield much lower expected sample inefficiencies for the na\"ive randomized sink (NRS) mechanism than the theoretical worst-case guarantee. The improvement ranges from roughly a factor of 5 (for a group size of 10) to almost 100 (for a group size of 210). This also indicates that the rate of decay of the inefficiency with the size of the group is faster than the theoretical guarantee. The experiment is run for different sizes of the agent groups with a random group of each size drawn multiple times from the dataset. The bars in \Cref{fig:real-nrs} shows the {\em average} expected sample inefficiencies of the mechanism with the standard deviations around them. Since NRS is a randomized mechanism, it is also worth looking at its worst-case performance on the same sampled datasets. \Cref{fig:real-worst} shows the {\em average} worst-case sample inefficiency along with its standard deviation. The line plot in both the figures shows the worst-case expected sample inefficiency for NRS for two alternatives.

\begin{figure}[h!]
\centering
\begin{subfigure}[b]{0.5\linewidth}
  \centering
  \includegraphics[width=\linewidth]{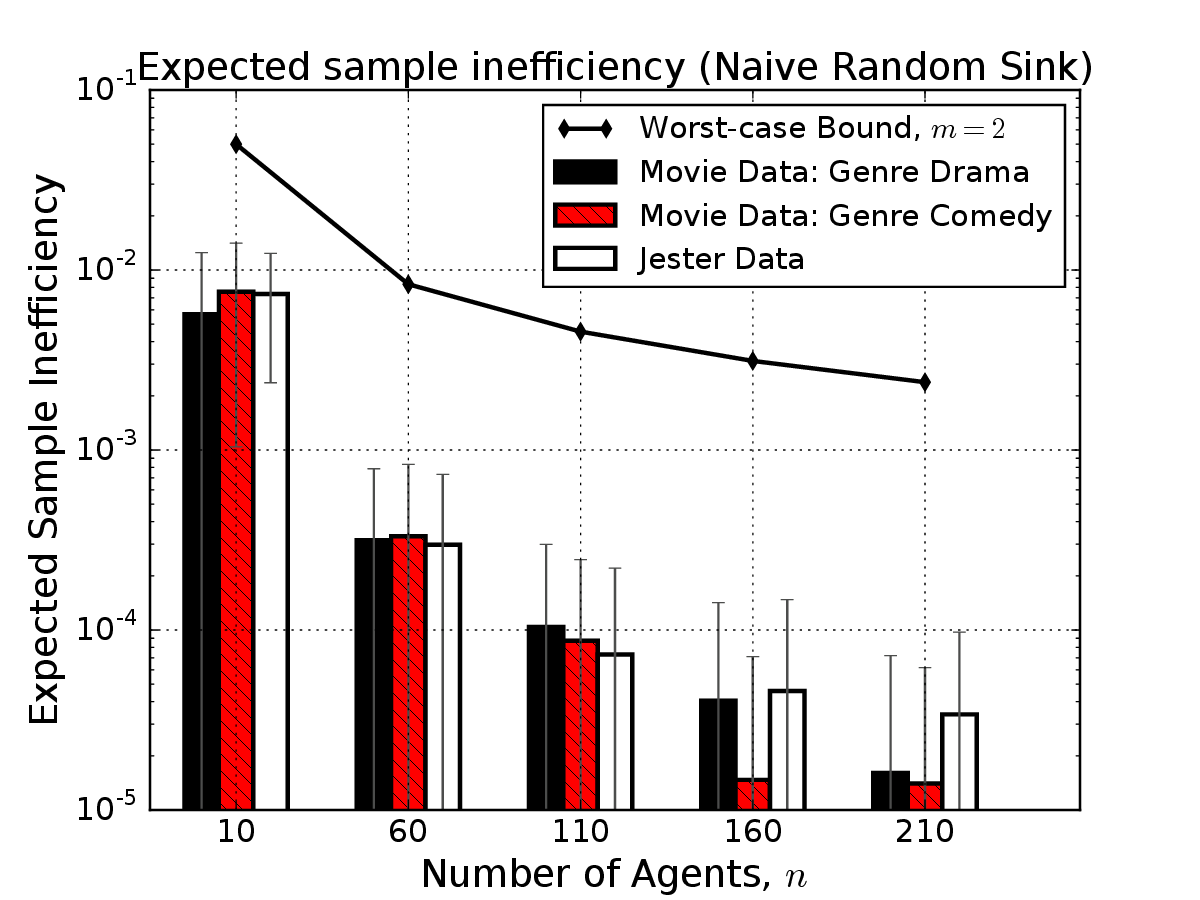}
  \caption{Na\"ive Random Sink mechanism}
  \label{fig:real-nrs}
\end{subfigure}%
\hspace{0cm}%
\begin{subfigure}[b]{0.5\linewidth}
  \centering
  \includegraphics[width=\linewidth]{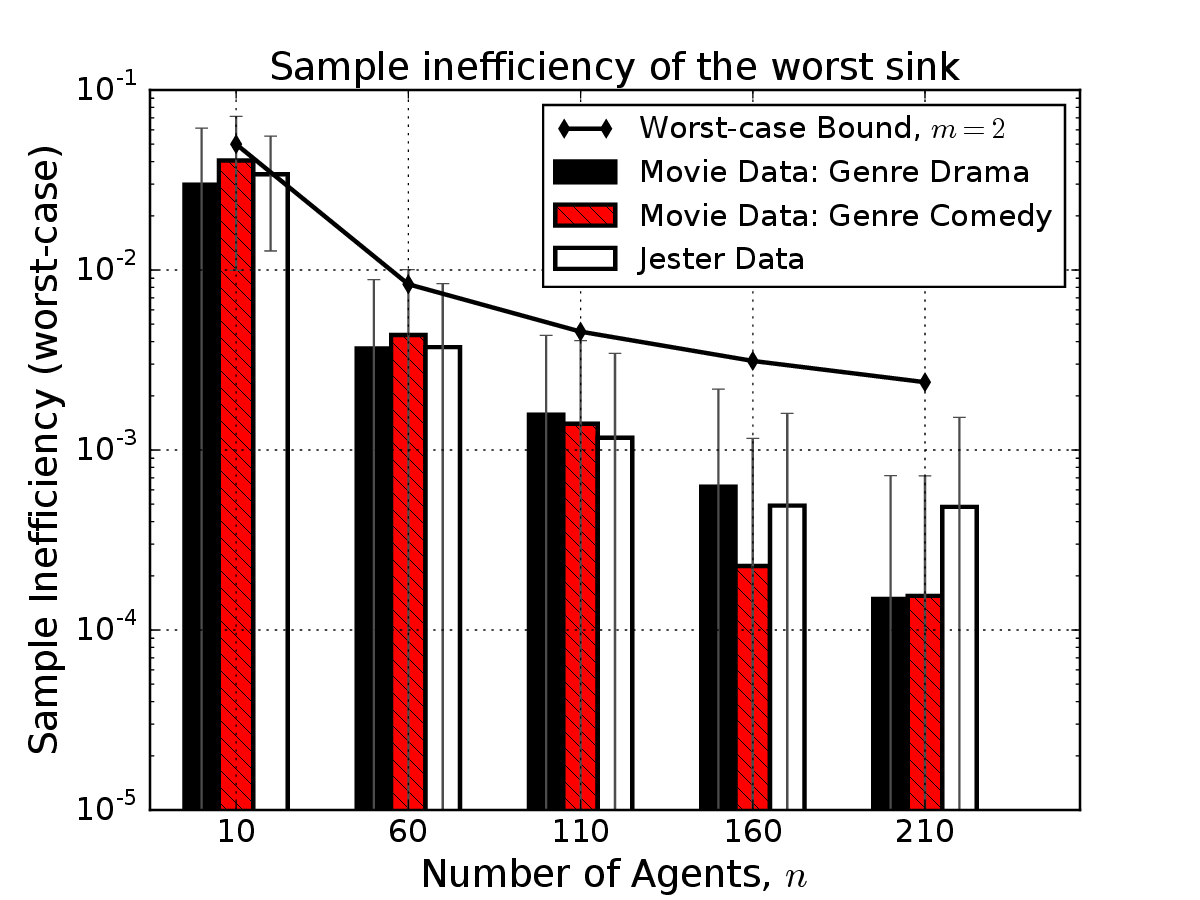}
  \caption{Worst-case behavior}
  \label{fig:real-worst}
\end{subfigure}
\caption{Performance of na\"ive randomized sink mechanism on MovieLens and Jester datasets.}
\label{fig:real-data}
\end{figure}
By the arguments following \Cref{thm:mod-irrel-sink}, we know that the inefficiency of the modified irrelevant sink (MIS) (\Cref{mech:mod-irrel-sink}) will be same as NRS. Since the MIS mechanism also picks exactly one sink, the worst-case behavior of a single sink illustrated above also applies to the MIS mechanism.

\section{Summary and future research}
\label{sec:concl}

In this paper, we provided several new results on the classic question of the interplay between efficiency and budget balance, properties that are incompatible with strategyproofness due to the Green-Laffont impossibility result, in the general quasi-linear setting.
We sought to understand the limits of minimal compromise between these two properties, both in the context of deterministic and randomized mechanism design framework.

We proved characterization results, and a tight lower bound for inefficiency, for deterministic budget-balanced mechanisms. We also proved that for unrestricted valuations, minimizing inefficiency and budget imbalance together does not provide any asymptotic advantage in the deterministic paradigm over requiring budget balance and minimizing inefficiency.

We proved that randomization helps---particularly when the number of alternatives is small compared to the number of agents. Motivated by our result for deterministic mechanisms that shows that a strategyproof, budget-balanced, neutral mechanism must include a sink, we introduced the class of {\em generalized sink mechanisms} which is a general (and adaptive to the valuation profile) way of picking sink agents. We showed that there exists strategyproof non-trivial mechanism (modified irrelevant sink) in this class that reduces the worst-case inefficiency by a factor of 2. We used an automated mechanism design approach for two agents and showed analytically that an optimal randomized mechanism offers further reduction in the inefficiency.

Experiments with real data from two applications compare the na\"ive randomized sink with its theoretical worst-case upper bound. We see that the mechanism perform well in practice and yield very little inefficiency ($\sim 1\%$ to $0.01\%$ depending on the group size).
This inefficiency is 5--100 times smaller than the worst case, and this relative difference increases with the number of agents.
We also consider the worst-case realization of this mechanism, and found that the sample inefficiency is close to the worst-case expected sample inefficiency of the mechanism.

Future research includes studying the structure of the optimal randomized mechanisms
that achieve the (theoretical) improved efficiency. Future work also includes investigating the rate of improvement of the optimal bound for a general number of agents.


\begin{thebibliography}{35}
\providecommand{\natexlab}[1]{#1}
\providecommand{\url}[1]{\texttt{#1}}
\expandafter\ifx\csname urlstyle\endcsname\relax
  \providecommand{\doi}[1]{doi: #1}\else
  \providecommand{\doi}{doi: \begingroup \urlstyle{rm}\Url}\fi

\bibitem[Arrow(1979)]{arrow1979property}
Kenneth Arrow.
\newblock \emph{The property rights doctrine and demand revelation under
  incomplete information}.
\newblock Economics and Human Welfare. New York Academic Press, 1979.

\bibitem[Ashlagi et~al.(2013)Ashlagi, Fischer, Kash, and
  Procaccia]{ashlagi2013mix}
Itai Ashlagi, Felix Fischer, Ian~A Kash, and Ariel~D Procaccia.
\newblock Mix and match: A strategyproof mechanism for multi-hospital kidney
  exchange.
\newblock \emph{Games and Economic Behavior}, pages 1--13, 2013.

\bibitem[Bailey(1997)]{bailey1997demand}
Martin~J Bailey.
\newblock The demand revealing process: to distribute the surplus.
\newblock \emph{Public Choice}, 91\penalty0 (2):\penalty0 107--126, 1997.

\bibitem[Cavallo(2006)]{cavallo2006optimal}
Ruggiero Cavallo.
\newblock Optimal decision-making with minimal waste: Strategyproof
  redistribution of {VCG} payments.
\newblock In \emph{Proceedings of the conference on autonomous agents and
  multiagent systems (AAMAS)}, pages 882--889, 2006.

\bibitem[Clarke(1971)]{Clarke1971}
Edward Clarke.
\newblock {Multipart Pricing of Public Goods}.
\newblock \emph{Public Choice}, \penalty0 (8):\penalty0 19--33, 1971.

\bibitem[Conitzer and Sandholm(2002)]{conitzer2002complexity}
Vincent Conitzer and Tuomas Sandholm.
\newblock Complexity of mechanism design.
\newblock In \emph{Proceedings of the conference on Uncertainty in Artificial
  Intelligence (UAI)}, pages 103--110. Morgan Kaufmann Publishers Inc., 2002.

\bibitem[d'Aspremont and G{\'e}rard-Varet(1979)]{dAspremont1979}
C.~d'Aspremont and L.~A. G{\'e}rard-Varet.
\newblock {Incentives and Incomplete Information}.
\newblock \emph{Journal of Public Economics}, 11(1):\penalty0 25--45, 1979.

\bibitem[Dobzinski et~al.(2006)Dobzinski, Nisan, and
  Schapira]{dobzinski2006truthful}
Shahar Dobzinski, Noam Nisan, and Michael Schapira.
\newblock Truthful randomized mechanisms for combinatorial auctions.
\newblock In \emph{Proceedings of the annual ACM Symposium on Theory of
  Computing}, pages 644--652. ACM, 2006.

\bibitem[Faltings(2004)]{faltings2004budget}
Boi Faltings.
\newblock A budget-balanced, incentive-compatible scheme for social choice.
\newblock In \emph{Agent-Mediated Electronic Commerce VI. Theories for and
  Engineering of Distributed Mechanisms and Systems}, pages 30--43. Springer,
  2004.

\bibitem[Feldman and Wilf(2011)]{feldman2011randomized}
Michal Feldman and Yoav Wilf.
\newblock Randomized strategyproof mechanisms for facility location and the
  mini-sum-of-squares objective.
\newblock \emph{arXiv preprint arXiv:1108.1762}, 2011.

\bibitem[Gary-Bobo and Jaaidane(2000)]{gary2000polling}
Robert~J Gary-Bobo and Touria Jaaidane.
\newblock Polling mechanisms and the demand revelation problem.
\newblock \emph{Journal of Public Economics}, 76\penalty0 (2):\penalty0
  203--238, 2000.

\bibitem[Goldberg et~al.(2001)Goldberg, Roeder, Gupta, and
  Perkins]{goldberg2001eigentaste}
Ken Goldberg, Theresa Roeder, Dhruv Gupta, and Chris Perkins.
\newblock Eigentaste: A constant time collaborative filtering algorithm.
\newblock \emph{Information Retrieval}, 4\penalty0 (2):\penalty0 133--151,
  2001.

\bibitem[Green and Laffont(1979)]{green1979incentives}
Jerry~R Green and Jean-Jacques Laffont.
\newblock \emph{Incentives in public decision making}.
\newblock North-Holland, 1979.

\bibitem[Groves(1973)]{Groves1973}
Theodore Groves.
\newblock {Incentives in Teams}.
\newblock \emph{Econometrica}, 41\penalty0 (4):\penalty0 617--31, July 1973.

\bibitem[Gujar and Narahari(2011)]{gujar2011redistribution}
Sujit~P Gujar and Yadati Narahari.
\newblock Redistribution mechanisms for assignment of heterogeneous objects.
\newblock \emph{Journal of Artificial Intelligence Research}, pages 131--154,
  2011.

\bibitem[Guo and Conitzer(2008)]{guo2008optimal}
Mingyu Guo and Vincent Conitzer.
\newblock Optimal-in-expectation redistribution mechanisms.
\newblock In \emph{Proceedings of the conference on autonomous agents and
  multiagent systems (AAMAS)}, pages 1047--1054, 2008.

\bibitem[Guo and Conitzer(2009)]{guo2009worst}
Mingyu Guo and Vincent Conitzer.
\newblock Worst-case optimal redistribution of VCG payments in multi-unit
  auctions.
\newblock \emph{Games and Economic Behavior}, 67\penalty0 (1):\penalty0 69--98,
  2009.

\bibitem[Guo and Conitzer(2014)]{guo2014better}
Mingyu Guo and Vincent Conitzer.
\newblock Better redistribution with inefficient allocation in multi-unit
  auctions.
\newblock \emph{Artificial Intelligence}, 216:\penalty0 287--308, 2014.

\bibitem[Gurobi(2015)]{gurobi}
Gurobi.
\newblock Gurobi optimizer reference manual, 2015.
\newblock URL \url{http://www.gurobi.com}.

\bibitem[Harper and Konstan(2016)]{harper2016movielens}
F~Maxwell Harper and Joseph~A Konstan.
\newblock The movielens datasets: History and context.
\newblock \emph{ACM Transactions on Interactive Intelligent Systems (TiiS)},
  5\penalty0 (4):\penalty0 19, 2016.

\bibitem[Heydenreich et~al.(2009)Heydenreich, M{\"u}ller, Uetz, and
  Vohra]{heydenreich2009characterization}
Birgit Heydenreich, Rudolf M{\"u}ller, Marc Uetz, and Rakesh~V Vohra.
\newblock Characterization of revenue equivalence.
\newblock \emph{Econometrica}, 77\penalty0 (1):\penalty0 307--316, 2009.

\bibitem[Krishna and Maenner(2001)]{krishna2001convex}
Vijay Krishna and Eliot Maenner.
\newblock Convex potentials with an application to mechanism design.
\newblock \emph{Econometrica}, 69\penalty0 (4):\penalty0 1113--1119, 2001.

\bibitem[Mass{\'o} et~al.(2015)Mass{\'o}, Nicol{\`o}, Sen, Sharma, and
  {\"U}lk{\"u}]{masso2015cost}
Jordi Mass{\'o}, Antonio Nicol{\`o}, Arunava Sen, Tridib Sharma, and Levent
  {\"U}lk{\"u}.
\newblock On cost sharing in the provision of a binary and excludable public
  good.
\newblock \emph{Journal of Economic Theory}, 155:\penalty0 30--49, 2015.

\bibitem[Mishra and Sen(2012)]{mishra2012roberts}
Debasis Mishra and Arunava Sen.
\newblock Roberts' theorem with neutrality: A social welfare ordering approach.
\newblock \emph{Games and Economic Behavior}, 75\penalty0 (1):\penalty0
  283--298, 2012.

\bibitem[Mishra and Sharma(2016)]{Mishra2016}
Debasis Mishra and Tridib Sharma.
\newblock On the optimality of the green-laffont mechanism.
\newblock Preliminary Draft, 2016.

\bibitem[Moulin(2009)]{moulin2009almost}
Herv{\'e} Moulin.
\newblock Almost budget-balanced vcg mechanisms to assign multiple objects.
\newblock \emph{Journal of Economic Theory}, 144\penalty0 (1):\penalty0
  96--119, 2009.

\bibitem[Moulin and Shenker(2001)]{moulin2001strategyproof}
Herv{\'e} Moulin and Scott Shenker.
\newblock Strategyproof sharing of submodular costs: budget balance versus
  efficiency.
\newblock \emph{Economic Theory}, 18\penalty0 (3):\penalty0 511--533, 2001.

\bibitem[Myerson(2013)]{myerson2013fundamentals}
Roger~B Myerson.
\newblock Fundamentals of social choice theory.
\newblock \emph{Quarterly Journal of Political Science}, 8\penalty0
  (3):\penalty0 305--337, 2013.

\bibitem[Procaccia and Tennenholtz(2009)]{procaccia2009approximate}
Ariel~D Procaccia and Moshe Tennenholtz.
\newblock Approximate mechanism design without money.
\newblock In \emph{Proceedings of the 10th ACM conference on Electronic
  Commerce (EC)}, pages 177--186. ACM, 2009.

\bibitem[Rockafellar(1997)]{rockafellar1997convex}
Ralph~Tyrrell Rockafellar.
\newblock \emph{Convex analysis}, volume~28.
\newblock Princeton university press, 1997.

\bibitem[Rothkopf(2007)]{ROTHKOPF07}
Michael~H. Rothkopf.
\newblock Thirteen reasons why the {Vickrey-Clarke-Groves} process is not
  practical.
\newblock \emph{Operations Research}, 55\penalty0 (2):\penalty0 191--197, 2007.

\bibitem[Roughgarden and Sundararajan(2009)]{roughgarden2009quantifying}
Tim Roughgarden and Mukund Sundararajan.
\newblock Quantifying inefficiency in cost-sharing mechanisms.
\newblock \emph{Journal of the ACM (JACM)}, 56\penalty0 (4):\penalty0 23, 2009.

\bibitem[Sandholm(2003)]{sandholm2003automated}
Tuomas Sandholm.
\newblock Automated mechanism design: A new application area for search
  algorithms.
\newblock In \emph{Principles and Practice of Constraint Programming--CP 2003},
  pages 19--36. Springer, 2003.

\bibitem[Thang(2010)]{thang2010randomized}
Nguyen~Kim Thang.
\newblock On randomized strategy-proof mechanisms without payment for facility
  location games.
\newblock In \emph{Proceedings of the Web and Internet Economics (WINE).
  Stanford, USA}, pages 13--16, 2010.

\bibitem[Vickrey(1961)]{Vickrey1961}
William Vickrey.
\newblock Counterspeculation, auctions, and competitive sealed tenders.
\newblock \emph{The Journal of Finance}, 16\penalty0 (1):\penalty0 8--37, 1961.

\end{thebibliography}




\pagebreak

\appendix

\section*{Appendix}
\setcounter{section}{1}

\subsection*{Proof of \Cref{thm:sink}}
\begin{proof}
 Consider the class of deterministic, strategyproof, and neutral mechanisms. \citet{mishra2012roberts} have shown that in the domain $V$, an allocation that satisfies the properties above must be a neutral affine maximizer (\Cref{def:affine-max}), that is, there exists $w_i \geq 0, \forall i \in N$, not all zero, such that,
 \begin{equation}
 \label{eq:affine-max}
  \text{$f(v) \in \argmax_{a \in A} \sum_{i \in N} w_i v_i(a).$}
 \end{equation}
 Additionally, the result by \citet{rockafellar1997convex} and \citet{krishna2001convex} states that for any convex type space, if the valuations are linear in type, then a strategyproof allocation satisfies revenue equivalence (\Cref{def:revenue-eq}). In our setting, the types of the agents are their valuations, which implies, trivially, that the valuations are linear in type. Also, they are drawn from the interval $(-\frac{M}{2}, \frac{M}{2})$, which is convex. So, revenue equivalence holds for the allocations in our setting. The following payment implements the affine maximizer allocation $f$ given by \Cref{eq:affine-max}:
 \begin{equation}
  \label{eq:payment-rev-eq}
  \text{$p_i(v_i,v_{-i}) = \left \{ \begin{array}{ll}
                               - \frac{1}{w_i} \left( \sum_{j \neq i} w_j v_j(f(v)) \right), & w_i > 0 \\
                               0, & w_i = 0
                              \end{array}
                              \right.$}
 \end{equation}
 for all $i \in N$. Since revenue equivalence holds in this setting, we conclude that any payment $\hat{p}_i, i \in N$ that makes $\langle f, \mathbf{\hat{p}} \rangle$ strategyproof, will be different from the above mentioned payments $\mathbf{p}$ by an additive factor $h_i(v_{-i})$ for each agent $i$ in every valuation profile.

 Now, we turn to proving the result of the theorem. We have the functional form of deterministic, strategyproof, neutral mechanisms given by \Cref{eq:affine-max}. If, on this class of mechanisms, we show that one cannot have weights $w_i > 0$ for all $i \in N$ while imposing budget balance, then we are done. This is because, if there exists one agent $i \in N$, for which $w_i = 0$, that agent is a {\em sink} agent as her valuations are never used by the social choice function and she is charged no payment.
 By revenue equivalence, any other payment that can implement the same allocation $f$ is $h_i(v_{-i})$. Putting this in the budget balance equation, we get $h_i(v_{-i}) = - \sum_{j \in N \setminus \{i\}} p_j(v)$, that is, she receives the payments made by the other agents. Thus agent $i$ is a sink agent. Hence, the proof is completed by proving the following claim.

 \begin{lemma}[Existence of $w_i=0$ Agent]
  \label{lemma:sink}
  A budget balanced mechanism $\langle f, \mathbf{p} \rangle$, where $f$ is a neutral affine maximizer on the domain $V$, must have at least one agent $i$ that has $w_i = 0$.
 \end{lemma}

 \begin{proof}
 Suppose for contradiction that $w_i > 0, \forall i \in N$.
 Since $f$ is a neutral affine maximizer (\Cref{eq:affine-max}) and revenue equivalence holds in $V$ (\Cref{eq:payment-rev-eq}), we know that the payments are of the form $p_i(v_i,v_{-i}) = h_i(v_{-i}) - \frac{1}{w_i} \left( \sum_{j \neq i} w_j v_j(f(v)) \right), \forall v \in V, \forall i \in N$.

Additionally, since the mechanism $\langle f, \mathbf{p} \rangle$ is also budget balanced, we have
\begin{align}
 \text{$\sum_{i=1}^n \left( h_i(v_{-i}) - \frac{1}{w_i} \left( \sum_{j \neq i} w_j v_j(f(v)) \right) \right) = 0, \ \forall v \in V$} \nonumber \\
 \text{$\Rightarrow \quad \sum_{i=1}^n h_i(v_{-i}) - \sum_{i=1}^n  \left( \sum_{j \neq i} \frac{1}{w_j} \right) w_i v_i(f(v)) = 0, \ \forall v \in V.$} \label{eq:bb-sink-lemma}
\end{align}

  For an easier exposition, we first explain the proof technique for $n=2$. Later the same proof is generalized to any number of agents.

  By assumption, $w_1, w_2 > 0$. Pick two valuation profiles $(v_1^+, v_2)$ and $(v_1^-,v_2)$ such that the affine maximizer alternative in the first is $a_1$ while that in the second is $a_2$, that is,
  \begin{align}
   w_1 v_1^+(a_1) + w_2 v_2(a_1) &> w_1 v_1^+(a_2) + w_2 v_2(a_2) \label{eq:ineq-1} \\
   w_1 v_1^-(a_1) + w_2 v_2(a_1) &< w_1 v_1^-(a_2) + w_2 v_2(a_2) \label{eq:ineq-2}
  \end{align}
 This can be done by choosing $v_1^+(a_2) = v_1^-(a_2) = v_1(a_2)$ (say) small and $v_2$ to be small enough for both alternatives, so that the valuation of agent $1$ for $a_1$ determines the resulting alternative of $f$. Therefore, the RHS of the inequalities above are the same. Since the inequality of \Cref{eq:ineq-2} is strict, let the difference of the RHS and LHS be $\delta > 0$. The allocations at these two profiles are: $f(v_1^+, v_2) = a_1$ and $f(v_1^-, v_2) = a_2$. Since the payments satisfy revenue equivalence and budget balance, \Cref{eq:bb-sink-lemma} holds, which gives
 \begin{align*}
  - \frac{1}{w_1} w_2 v_2(a_1) + h_1(v_2) - \frac{1}{w_2} w_1 v_1^+(a_1) + h_2(v_1^+) &= 0 \\
  - \frac{1}{w_1} w_2 v_2(a_2) + h_1(v_2) - \frac{1}{w_2} w_1 v_1^-(a_2) + h_2(v_1^-) &= 0.
 \end{align*}
 Subtracting the first equation from the second and rearranging, we get
 \begin{equation}
  \label{eq:bb-equation}
  \frac{1}{w_1} w_2 (v_2(a_1) - v_2(a_2)) = \frac{1}{w_2} w_1 (v_1^-(a_2) - v_1^+(a_1)) - h_2(v_1^-) + h_2(v_1^+).
 \end{equation}
 Note that the RHS is independent of $v_2$. Therefore, if $v_2(a_1)$ is increased by a small amount ($< \delta / w_2$), both the inequalities given by \Cref{eq:ineq-1,eq:ineq-2} still hold, but \Cref{eq:bb-equation} fails to hold, which is a contradiction.

 The general proof of this lemma extends this idea to any number of agents $n \geq 2$.
 We prove this for a set of alternatives $A = \{0,1\}$. Consider this setting as that of a public project. In alternative $0$, the project is not undertaken, yielding every agent a value of zero, and when $1$ is chosen---i.e., the project is undertaken---the valuation of each agent is denoted by a single real number. This assumption helps us reduce the notational complexity. The proof, however, is completely general for any number of alternatives.

 Let the agents be numbered in decreasing order of their weights WLOG, that is, $w_i \geq w_{i+1}, i=1,2, \ldots, n-1$. We consider the following valuation profile: $(v_1 + \delta, v_2 + \delta, \ldots, v_{n-1} + \delta, v_n), \delta > 0$ such that
\begin{equation}
 \label{eq:sink-ineq}
 -\delta \sum_{i=1}^{n-1} w_i < \sum_{i=1}^{n} w_i v_i < -\delta \sum_{i=1}^{n-2} w_i
\end{equation}
The above inequalities imply that the affine maximizer alternative given by \Cref{eq:affine-max} for the profile mentioned above is $1$. However, if any agent $i$'s, $i = 1,2,\ldots,n-1$, valuation changes from $v_i + \delta$ to $v_i$, the alternative changes to $0$. We use a generic notation $v^S$ to denote this profile, where $S$ denotes the set of agents such that the valuations of all the agents $k \in S$ are $v_k$, and the valuations of the agents $j \notin S$ are $v_j + \delta$. Hence, $v^{\{n\}}$ is the profile mentioned before: $(v_1 + \delta, v_2 + \delta, \ldots, v_{n-1} + \delta, v_n)$ and $v^{\{n-1,n\}}$ is the profile: $(v_1 + \delta, v_2 + \delta, \ldots, v_{n-1}, v_n)$, for example.

Since, $f(v^{\{n\}}) = 1$, from \Cref{eq:bb-sink-lemma} we have,
\begin{equation}
 \label{eq:bb-mod-1}
 \left( \sum_{i=1}^{n-1} h_i(v^{\{n\}}_{-i}) + h_n(v^{\{n\}}_{-n}) \right) - \left( \sum_{i=1}^n  \left( \sum_{j \neq i} \frac{1}{w_j} \right) w_i v_i + \sum_{i=1}^{n-1}  \left( \sum_{j \neq i} \frac{1}{w_j} \right) w_i \delta \right) = 0.
\end{equation}
The idea of the proof is to make a series of substitutions in the first parentheses of the expression above, leaving the terms in the other parentheses unchanged. Note that, the expression in the second parentheses depends on $v_n$, while the expression $h_n(v^{\{n\}}_{-n})$ does not. The substitutions sequentially eliminate the dependency on $v_n$ from all the terms in the first parentheses, similar to what we did in the two agent case before. This leads to a contradiction, since $v_n$ can be perturbed
arbitrarily small so that it continues to satisfy the inequalities of \Cref{eq:sink-ineq}, our only assumed condition, but violates the equality in \Cref{eq:bb-mod-1}.

The substitutions will involve the term $\sum_{i=1}^{n-1} h_i(v^{\{n\}}_{-i})$ in the first parentheses of \Cref{eq:bb-mod-1}. Consider the profiles $v^{\{j,n\}}, j = 1, \ldots, n-1$. In each of these profiles, $f(v^{\{j,n\}}) = 0$ (due to the choice of $v^{\{n\}}$ in \Cref{eq:sink-ineq}). Hence,
\begin{equation}
 \label{eq:bb-mod-2}
 \sum_{i=1}^{n-1} h_i(v^{\{j,n\}}_{-i}) + h_n(v^{\{j,n\}}_{-n}) = 0, \ \forall j \in \{1,\ldots, n-1\}.
\end{equation}
Note that $v^{\{i,n\}}_{-i} = v^{\{n\}}_{-i}$. Hence, we can substitute terms from \Cref{eq:bb-mod-2} to the terms in the first parentheses of \Cref{eq:bb-mod-1} to get,
\begin{gather}
 \label{eq:bb-mod-3}
 \left( - \sum_{i=1}^{n-1} \sum_{j \neq \{i,n\}} h_j(v^{\{i,n\}}_{-j}) - \sum_{j \neq n} h_n(v^{\{j,n\}}_{-n}) + h_n(v^{\{n\}}_{-n}) \right) - \left( \sum_{i=1}^n  \left( \sum_{j \neq i} \frac{1}{w_j} \right) w_i v_i + \sum_{i=1}^{n-1}  \left( \sum_{j \neq i} \frac{1}{w_j} \right) w_i \delta \right) = 0.
\end{gather}
We continue replacing the terms $h_j(v^{\{i,n\}}_{-j})$ in the first summation of the first parentheses above. All other terms in that parentheses are $h_n$ functions and, therefore, are independent of $v_n$. For every $i \neq n$, consider the valuation profiles $v^{\{j,i,n\}}, j \neq i, n$. By \Cref{eq:sink-ineq}, $f(v^{\{j,i,n\}}) = 0$, so we get an equality similar to \Cref{eq:bb-mod-2}: $\sum_{k=1}^{n-1} h_k(v^{\{j,i,n\}}_{-k}) + h_n(v^{\{j,i,n\}}_{-n}) = 0, \ \forall j \neq i,n.$
Also, $v^{\{j,i,n\}}_{-j} = v^{\{i,n\}}_{-j}$. So, we follow the same procedure to replace the terms $h_j(v^{\{i,n\}}_{-j})$ in \Cref{eq:bb-mod-3} to yield a similar equality where more terms that were dependent on $v_n$ are now replaced with $h_n$ functions, which are independent of $v_n$. Since the number of agents is finite, this process will stop after a finite number of iterations, reducing the terms in the first parentheses
consisting only of $h_n$ functions. This construction shows that a small perturbation of $v_n$, which keeps \Cref{eq:sink-ineq} unaffected, will violate the equality obtained through the iterative procedure described above. This completes the proof of the lemma.
 \end{proof}
 \Cref{lemma:sink} shows that there exists an agent with weight zero, which is a sink agent, and hence the proof of Theorem~\ref{thm:sink} is complete.
\end{proof}

\subsection*{Proof of \Cref{thm:det-bound}}
\begin{proof}
From \Cref{thm:sink}, we know that any $f$ that satisfies the properties mentioned in the statement of the current theorem
must be a neutral affine maximizer with at least one agent $i^*$ that has $w_{i^*} = 0$. We now show that the minimum sample inefficiency $r_n^M(f)$ is achieved when there is exactly {\em one} such agent $i^*$ and the weights of the other agents $i \in N \setminus \{i^*\}$ are equal. 
We prove this in two steps: (a) first we show that the proposed single-sink mechanism indeed gives a worst-case sample inefficiency of $1/n$, and having multiple sinks can only make the inefficiency worse, (b) having unequal weights for the rest of the agents is suboptimal for sample inefficiency.
This immediately proves the theorem since the proposed mechanism has sample inefficiency $\frac{1}{n}$. 

Step (a): the proposed mechanism, having exactly one $i^*$ such that $w_{i^*} = 0$, picks the welfare maximizing allocation not considering the sink agent $i^*$,
that is, $f(v) \in \argmax_{a \in A} \sum_{j \in N \setminus \{i^*\}} v_j(a)$. Denoting $a^*(v)$ to be the efficient allocation, we write:
 \begin{gather}
 \begin{split}
  \lefteqn{\max_{a \in A} \sum_{i \in N} v_i(a) - \sum_{i \in N} v_i(f(v))= \sum_{i \in N} v_i(a^*(v)) - \sum_{i \in N} v_i(f(v))} \\
  &= v_{i^*}(a^*(v)) - v_{i^*}(f(v)) + \left[\sum_{j \in N \setminus \{i^*\}} v_j(a^*(v)) - \sum_{j \in N \setminus \{i^*\}} v_j(f(v)) \right] < \left( \frac{M}{2} - \left( -\frac{M}{2} \right) \right) + 0 = M
 \end{split}
 \label{eq:inefficiency-bound}
 \end{gather}
 The first part of the above inequality comes from the fact that the valuations are drawn from $(-\frac{M}{2},\frac{M}{2})$ so the difference in valuation can at most be $M$. The second part of the inequality holds because $f(v)$ is the welfare maximizing allocation excluding agent $i^*$.

 It is easy to verify that this inequality is tight at the following valuation profile: $v_{i^*}(a) = \frac{M}{2} - \delta, v_{i^*}(z) = -\frac{M}{2} + \gamma, \ \forall z \neq a$, and $v_j(b) = -\frac{M}{2} + \epsilon, v_j(z) = -\frac{M}{2} + \frac{\epsilon}{2}, \ \forall z \neq b, \forall j \neq i^*$, where $\delta, \gamma, \epsilon >0$ are arbitrarily small. The alternatives are: $a^*(v) = a, f(v) = b$. Clearly, this satisfies the above inequality and by taking $\delta, \gamma, \epsilon \to 0$, we get that the supremum of the difference term approaches $M$, and hence the sample inefficiency becomes $\frac{1}{n}$.

 This counterexample also shows that having more than one sink agent will make the sample inefficiency worse---that is, larger. This is because we can replicate the valuation of $i^*$ for every other sink and the inequality above will be tightly upper bounded at $2M$ for 2 sinks, $3M$ for 3 sinks, etc. Consequently, the sample inefficiency increases to $\frac{2}{n}$, $\frac{3}{n}$ etc.

 To prove step (b), we use the following lemma.
 \begin{lemma}[Lowest Sample Inefficiency]
  \label{lemma:lub}
  In the class of neutral affine maximizers given by \Cref{eq:affine-max} having a sink agent $i^*$ (i.e., $w_{i^*} = 0$), the lowest sample inefficiency is achieved when $w_i = w$ for all $i \in N \setminus \{i^*\}$.
 \end{lemma}

 \begin{proof}
  Suppose not, that is, $\exists j,j' \in N \setminus \{i^*\}$ such that, WLOG, $w_j > w_{j'}$. Consider the following valuation profile:
  \begin{gather*}
  \begin{array}{rcl}
   v_{i^*}(a) & = & -\frac{M}{2} + \gamma, \\
   v_{j}(a) & = & \frac{M}{2} - \delta, \\
   v_{j'}(a) & = & -\frac{M}{2} + \gamma,
  \end{array}
  \hspace{1mm}
  \begin{array}{rcl}
   v_{i^*}(b) & = & \frac{M}{2} - \delta, \\
   v_{j}(b) & = & \frac{M}{2} - \frac{w_{j'}}{w_j} M - \epsilon, \\
   v_{j'}(b) & = & \frac{M}{2} - \delta, \\
  \end{array}
  \hspace{1mm}
  \begin{array}{rcl}
   v_i(a) & = & v_i(b), \ \forall i \in N \setminus \{ i^*, j, j'\}, \\
   v_i(z) & = & -\frac{M}{2} + \frac{\gamma}{2}, \ \forall z \in A \setminus \{a,b\}, \forall i \in N.
  \end{array}
  \end{gather*}
  The constants $\delta, \gamma, \epsilon>0$ are arbitrarily small.
  It is easy to verify that on this profile, by choosing the constants $\delta, \gamma, \epsilon$ appropriately small, the affine maximizer will return $a$. But the efficient alternative is $a^*(v) = b$. Consider the term $\sum_{i \in N} v_i(a^*(v)) - \sum_{i \in N} v_i(f(v))$. To this term, agent $i^*$ contributes inefficiency $M$, agent $j'$ contribute $M$, and agent $j$ contributes $\left(-\frac{w_{j'}}{w_j} M \right)$, taking the limiting values of $\delta, \gamma, \epsilon \to 0$. Therefore, the inefficiency term on this profile equals $M + \left( 1 - \frac{w_{j'}}{w_j} \right) M > M$, the maximum inefficiency when the weights are equal except $i^*$ (by the arguments just before \Cref{lemma:lub}). Hence, this mechanism cannot achieve the lowest sample inefficiency, which is a contradiction.
 \end{proof}
 \Cref{lemma:lub} and the arguments before it complete the proof of the theorem.
\end{proof}

\subsection*{Proof of \Cref{thm:unimprovability}} \smallskip

\begin{proof}
 Suppose, there exists a deterministic, strategyproof, and neutral mechanism $\langle f, \mathbf{p} \rangle$ that also satisfies $\lim_{n \to \infty} \sup_{v \in V} T_2^n(\mathbf{p},v) = 0$. It implies that, at the limit, the mechanism has no budget imbalance, i.e., $\lim_{n \to \infty} \sup_{v \in V} \left| \sum_{i=1}^n p_i(v) \right| = 0$. From the arguments in \Cref{thm:sink} (\Cref{eq:affine-max,eq:payment-rev-eq}), we know that $f$ is a neutral affine maximizer and payments are of the form $p_i(v_i,v_{-i}) = h_i(v_{-i}) - \frac{1}{w_i} \left( \sum_{j \neq i} w_j v_j(f(v)) \right), \forall w_i > 0$. We already have the sink mechanism where at least one $w_i = 0$ and the above sum can be made smallest (exactly zero) for every profile $v \in V$. However, that yields a constant upper bound for the term $\lambda \cdot T_1^n(f,v) + (1 - \lambda) \cdot T_2^n(\mathbf{p},v)$. Hence, we need to consider the case $w_i > 0, \forall i$, which implies that
\[\lim_{n \to \infty} \sup_{v \in V} \left| \sum_{i=1}^n \left( h_i(v_{-i}) - \frac{1}{w_i} \left( \sum_{j \neq i} w_j v_j(f(v)) \right) \right) \right| = 0.\]
This implies that for every $\epsilon > 0$, there exists $N_{\epsilon} \in \mathbb{Z}_{\geq 0}$ such that for all $n \geq N_{\epsilon}$,
\begin{equation}
 \label{eq:convergence}
 \left| \sum_{i=1}^n \left( h_i(v_{-i}) - \frac{1}{w_i} \left( \sum_{j \neq i} w_j v_j(f(v)) \right) \right) \right| < \epsilon, \ \forall v \in V.
\end{equation}
We show that this identity leads to a contradiction for an appropriately chosen $v$. Note that this immediately proves the theorem, because if there does not exist any mechanism $\langle f, \mathbf{p} \rangle$ that satisfies the properties mentioned in the theorem statement and makes $\lim_{n \to \infty} \sup_{v \in V} T_2^n(\mathbf{p},v) = 0$, then the best possible lower bound is a constant, i.e., $\sup_{v \in V} T_2^n(\mathbf{p},v) = \Omega(1)$. Therefore, the best lower bound for the spillover factor $\rho_n(f, \mathbf{p})$ is $\Omega \left( \frac{1}{n} \right)$ and this is achievable by the sink mechanism.

We prove this for a set of alternatives $A = \{0,1\}$ for similar reasons mentioned in \Cref{lemma:sink}.
As an illustration of the general proof, let us consider the same argument when $N_{\epsilon} = 2$. Let the valuations are $(v_1 + \delta, v_2)$ for alternative $1$ and zero otherwise ($\delta > 0$). Also assume that the numbers are such that
\[w_1 (v_1 + \delta) + w_2 v_2 > 0, \text { and } w_1 v_1 + w_2 v_2 < 0.\]
That means, the affine maximizer results in $1$ at profile $(v_1 + \delta, v_2)$ and $0$ at $(v_1, v_2)$. The above inequalities can be written concisely as
\begin{equation}
 \label{eq:af-max-ineq}
 -w_1 \delta < w_1 v_1 + w_2 v_2 < 0.
\end{equation}
Now by the convergence relation of \Cref{eq:convergence}, we have
\begin{align*}
 &\left| h_1(v_2) + h_2(v_1 + \delta) - \frac{w_2}{w_1} v_2 - \frac{w_1}{w_2} (v_1 + \delta) \right| < \epsilon \\
 &\left| h_1(v_2) + h_2(v_1) \right| < \epsilon
\end{align*}
These inequalities imply\footnote{If $|x+z| < \epsilon$ and $|y+z| < \epsilon$, then $|x-y| = |x+z - (y+z)| \leq |x+z| + |y+z| < 2 \epsilon$.}
\begin{align}
 \left| h_2(v_1 + \delta) - \frac{w_2}{w_1} v_2 - \frac{w_1}{w_2} (v_1 + \delta) - h_2(v_1) \right| &< 2\epsilon \nonumber \\
 \Rightarrow \quad \left| \frac{w_2}{w_1} v_2 + \left(h_2(v_1) - h_2(v_1 + \delta) + \frac{w_1}{w_2} (v_1 + \delta)\right) \right| &< 2\epsilon \label{eq:value-2agents}
\end{align}
But this inequality is violated by choosing a large enough $\delta$ and large negative $v_2$ in \Cref{eq:af-max-ineq}. This is possible to pick since the valuations are picked from $(-M/2, M/2)$ and $M$ is large by definition of $\rho_n$ (\Cref{eq:sample-joint-inefficiency}). Also note that, the term within parentheses in \Cref{eq:value-2agents} is independent of $v_2$, hence changes in the $v_2$ will not affect them. Our only assumed relation is \Cref{eq:af-max-ineq}, and a suitable choice satisfying it violates \Cref{eq:value-2agents}.

The general proof of this theorem extends this idea for any $N_{\epsilon} = n \geq 2$. Let the agents are numbered in the decreasing order of their weights WLOG, i.e., $w_i \geq w_{i+1}, i=1,2, \ldots, n-1$. We consider the valuation profile $(v_1 + \delta, v_2 + \delta, \ldots, v_{n-1} + \delta, v_n), \delta > 0$ such that
\begin{equation}
 \label{eq:gen-unimprovability-ineq}
 -\delta \sum_{i=1}^{n-1} w_i < \sum_{i=1}^{n} w_i v_i < -\delta \sum_{i=1}^{n-2} w_i
\end{equation}
The above inequalities imply that the affine maximizer alternative for the profile mentioned above is $1$. However, if any agent $i$'s, $i = 1,2,\ldots,n-1$, valuation changes to $v_i$ from $v_i + \delta$, the alternative changes to $0$. We use a generic notation $v^S$ to denote this profile, where $S$ denotes the set of agents such that for all agents $k \in S$, the valuations are $v_k$, and for all $j \notin S$, the valuations are $v_j + \delta$. Hence, $v^{\{n\}}$ is the profile mentioned before: $(v_1 + \delta, v_2 + \delta, \ldots, v_{n-1} + \delta, v_n)$ and $v^{\{n-1,n\}}$ is the profile: $(v_1 + \delta, v_2 + \delta, \ldots, v_{n-1}, v_n)$, for example. Note that, the following term in \Cref{eq:convergence} can be reorganized as
\[\sum_{i=1}^n \frac{1}{w_i} \left( \sum_{j \neq i} w_j v_j(f(v)) \right) = \sum_{i=1}^n  \left( \sum_{j \neq i} \frac{1}{w_j} \right) w_i v_i(f(v)).\]
Since, $f(v^{\{n\}}) = 1$, from \Cref{eq:convergence} we have
\begin{equation}
 \label{eq:convergence-mod-1}
 \left| \left( \sum_{i=1}^{n-1} h_i(v^{\{n\}}_{-i}) + h_n(v^{\{n\}}_{-n}) \right) - \left( \sum_{i=1}^n  \left( \sum_{j \neq i} \frac{1}{w_j} \right) w_i v_i + \sum_{i=1}^{n-1}  \left( \sum_{j \neq i} \frac{1}{w_j} \right) w_i \delta \right) \right| < \epsilon.
\end{equation}
The idea of the proof is to make a series of substitutions in the first parentheses of the expression above, leaving the terms in the other parentheses unchanged. Note that, the expression in the second parentheses depends on $v_n$, while the expression $h_n(v^{\{n\}}_{-n})$ does not. The substitutions sequentially eliminate the dependency on $v_n$ from all the terms in the first parentheses, similar to what we did in the two agent case before. This will also increase the RHS of the inequality in \Cref{eq:convergence-mod-1}, but it will be a finite constant factor of $\epsilon$. This leads to a contradiction, since $v_n$ can be chosen arbitrarily large negative by choosing a large positive $\delta$, and still continues to satisfy \Cref{eq:gen-unimprovability-ineq} but violates \Cref{eq:convergence-mod-1}.

The substitutions will involve the term $\sum_{i=1}^{n-1} h_i(v^{\{n\}}_{-i})$ in the first parentheses of \Cref{eq:convergence-mod-1}. Consider the profiles $v^{\{j,n\}}, j = 1, \ldots, n-1$. In each of these profiles, $f(v^{\{j,n\}}) = 0$ (due to the choice of $v^{\{n\}}$ in \Cref{eq:gen-unimprovability-ineq}). Hence,
\begin{equation}
 \label{eq:convergence-mod-2}
 \left| \sum_{i=1}^{n-1} h_i(v^{\{j,n\}}_{-i}) + h_n(v^{\{j,n\}}_{-n}) \right| < \epsilon, \ \forall j \in \{1,\ldots, n-1\}.
\end{equation}
Note that $v^{\{i,n\}}_{-i} = v^{\{n\}}_{-i}$. Hence, we can substitute terms from \Cref{eq:convergence-mod-2} to the terms in the first parentheses of \Cref{eq:convergence-mod-1} to get
\begin{align}
 \label{eq:convergence-mod-3}
 \lefteqn{ \left| \left( - \sum_{i=1}^{n-1} \sum_{j \neq \{i,n\}} h_j(v^{\{i,n\}}_{-j}) - \sum_{j \neq n} h_n(v^{\{j,n\}}_{-n}) + h_n(v^{\{n\}}_{-n}) \right) - \right. } \nonumber \\
 & \qquad \left. \left( \sum_{i=1}^n  \left( \sum_{j \neq i} \frac{1}{w_j} \right) w_i v_i + \sum_{i=1}^{n-1}  \left( \sum_{j \neq i} \frac{1}{w_j} \right) w_i \delta \right) \right| < n \epsilon.
\end{align}
We now replace the terms $h_j(v^{\{i,n\}}_{-j})$ in the first summation of the first parentheses above. All other terms in that parentheses are $h_n$ functions and, therefore, are independent of $v_n$. For every $i \neq n$, consider the valuation profiles $v^{\{j,i,n\}}, j \neq i, n$. By \Cref{eq:gen-unimprovability-ineq}, $f(v^{\{j,i,n\}}) = 0$, hence, we get an inequality similar to \Cref{eq:convergence-mod-2}:
\begin{equation}
 \label{eq:convergence-mod-4}
 \left| \sum_{k=1}^{n-1} h_k(v^{\{j,i,n\}}_{-k}) + h_n(v^{\{j,i,n\}}_{-n}) \right| < \epsilon, \ \forall j \neq i,n.
\end{equation}
Also, note that, $v^{\{j,i,n\}}_{-j} = v^{\{i,n\}}_{-j}$. So, we follow the same procedure to replace the terms $h_j(v^{\{i,n\}}_{-j})$ in \Cref{eq:convergence-mod-3} to yield a similar inequality where the RHS is replaced by a larger term. Since the number of agents is finite, this process will stop after a finite number of iterations, reducing the terms in the first parentheses only consisting of $h_n$ functions, which are independent of $v_n$, and the RHS of the inequality being a finite factor $K(n) \epsilon$ (say). This construction shows that the choice of a suitably large $\delta$ and negative $v_n$, which keeps \Cref{eq:gen-unimprovability-ineq} unaffected, can violate the inequality obtained through the iterative procedure described above. Hence the claim is established.
\end{proof}

\subsection*{Proof of \Cref{thm:gen-sink}}

\begin{proof}
 Assume $m = n+1$. The proof generalizes to any $m > n$. Consider the valuation profile $v_i = (v_i(a_1), \ldots, v_i(a_n), v_i(a_{n+1}))$ where $v_i(a_i) = -M/2 + \epsilon/2, v_i(a_{n+1}) = M/2 - \epsilon$ and $v_i(a_j) = M/2 - \epsilon/2, \forall j \neq i, n+1$, and $\epsilon > 0$ is arbitrarily small, $i \in N$. This profile is possible to construct since $m > n$. Clearly, the efficient alternative is $a_{n+1}$, but if any agent $i$ is picked as a sink, the alternative changes to $a_i$, which has inefficiency of $M - \epsilon$. Therefore the expected sample inefficiency for any generalized sink mechanism is $\frac{1}{nM} (M - \epsilon)$. Taking $\epsilon \to 0$ proves the theorem.
\end{proof}

\subsection*{Proof of \Cref{thm:incr-inefficiency}}

\begin{proof}
 Suppose, for $m_2$ alternatives the valuation profile $v^*$ yields the worst inefficiency $r_{n,m_2}^M(f)$. Clearly, we can append the other alternatives when we increase the number of alternatives to $m_1$ with values arbitrarily close to $-M/2$ so that they never change the optimal alternative. Hence the inefficiency cannot decrease.
\end{proof}

\subsection*{Proof of \Cref{thm:naive-random}}

\begin{proof}
 Consider an arbitrary agent $i$. If agent $i$ is chosen as a sink, the maximum {\em absolute} inefficiency that it can produce is $M$ (by the same argument as in \Cref{eq:inefficiency-bound}, and we refer to the unnormalized difference term $\sum_{i \in N} v_i(a^*(v)) - \sum_{i \in N} v_i(f(v))$ as the absolute inefficiency). Say the efficient alternative is $a$. This inefficiency is achieved when the sum of valuations of the agents other than $i$ at the other alternative $b$ is just higher than that of those agents at $a$, i.e., $\sum_{j \neq i} v_j(b) = \sum_{j \neq i} v_j(a) + \epsilon$, where $\epsilon > 0$ and small, and also the difference in valuations of agent $i$ at these two alternatives is maximum, i.e., $v_i(a) - v_i(b) = M - \delta$, where $\delta > 0$ and small. This implies that, without $i$, the population is almost equally divided among the alternatives $a$ and $b$ with a marginal bias to $b$ and agent $i$ is `maximally' in favor of $a$. The difference $v_i(a) - v_i(b) = M - \delta$ is achieved only when the values are close to $v_i(a) = \frac{M}{2} - \frac{\delta}{2}$ and $v_i(b) = -\frac{M}{2} + \frac{\delta}{2}$, since all valuations must lie within $\left(\frac{M}{2}, -\frac{M}{2} \right)$. Agent $i$, therefore, is {\em pivotal} for making the decision in favor of $a$.

 Now, we claim that there cannot be more than $\left \lceil \frac{n}{2} \right \rceil$ such pivotal agents $i$. Suppose for contradiction that there are $> \left \lceil \frac{n}{2} \right \rceil$ such pivotal agents. We present the argument for $\left \lceil \frac{n}{2} \right \rceil + 1$ for brevity, but it generalizes to the case with more number of pivotal agents.
 For each of these agents $k$, $v_k(a)$ is arbitrarily close to $\frac{M}{2}$ and $v_k(b)$ is arbitrarily close to $-\frac{M}{2}$. Therefore, if any of them, say agent $k^*$,
 is chosen as a sink, there are $\left \lceil \frac{n}{2} \right \rceil$ other similar agents who still make the decision of the mechanism in favor of $a$ (since the sum valuation for $a$ will be larger than $b$ by an unsurmountable value $\approx \left \lceil \frac{n}{2} \right \rceil M$) irrespective of the valuation profiles of the other agents. This implies that $k^*$ is not a pivotal agent, which is a contradiction.

 The mechanism chooses each sink agent with probability $\frac{1}{n}$. Therefore, the expected inefficiency can at most be $\frac{1}{n} \cdot \left \lceil \frac{n}{2} \right \rceil \cdot M$. Divide this by $nM$ to get the expected sample inefficiency. It is easy to see that this bound is tight. The valuation profile that achieves this bound has $\left \lceil \frac{n}{2} \right \rceil$ agents having valuations $v(a) \approx M/2, v(b) \approx -M/2$ and the rest of the agents have the reverse valuations.
\end{proof}

\subsection*{Proof of \Cref{thm:LB-gen-sink}}

\begin{proof}
 For $n=m=2$, either the two agents prefer the same candidate or opposing candidates. (Ties are broken in favor of $a_1$, say.) If they agree, any probability of the generalized sink mechanism yields zero absolute inefficiency, since the efficient alternative will be chosen irrespective of which agent is the sink. Note that the payments are always zero for two-agent generalized sink mechanisms. When the agents oppose, we claim that in every opposing profile, a strategyproof generalized sink mechanism must have the same probability of picking the sinks. Suppose not, that is, for a specific generalized sink mechanism $g: V \to \Delta N$, the probabilities of picking the sinks are different in profiles $v = (v_1, v_2)$ and $v' = (v_1', v_2')$, i.e., $g(v) \neq g(v')$. Consider the transition: $v = (v_1, v_2) \to (v_1', v_2) \to (v_1', v_2') = v'$. The sink-picking probabilities $g$ must have changed in at least one of these two transitions, that is, either $g(v_1, v_2) \neq g(v_1', v_2)$ or $g(v_1', v_2) \neq g(v_1', v_2')$.
 But this is a contradiction to strategyproofness since at least one agent will misreport in that profile pair. She will prefer to increase the probability of the other agent becoming sink so that her favorite candidate has higher probability of being selected, which increases her utility since payment is zero. For example, suppose in the first transition, the probability of agent $1$ being sink is higher in the profile $(v_1,v_2)$, and consequently, probability of agent $2$ being sink is lower. Then agent $1$ will misreport her valuation to $v_1'$. Now, among all fixed probability distributions, $(0.5, 0.5)$ gives the minimum absolute inefficiency which is $\frac{1}{2}$.
\end{proof}

\subsection*{Proof of \Cref{thm:random-opt-lower}} \smallskip

The problem of \Cref{eq:optimization} can be written in a standard LP representation as follows. The notation $f_a(v)$ denotes the probability of picking alternative $a$ by the randomized mechanism $f$ when the valuation profile is $v$.

\begin{equation}
 \label{eq:LP}
 \underset{f,\mathbf{p}}{\min} \qquad \ell
\end{equation}
\begin{flalign*}
&\text{s.t.} &&\\
 & [ v_1(a) \cdot f_a(v_1,v_2) + v_1(b) \cdot f_b(v_1,v_2) - p_1(v_1,v_2) ] - [ v_1(a) \cdot f_a(v_1',v_2)  \\
  & \qquad + v_1(b) \cdot f_b(v_1',v_2) - p_1(v_1',v_2) ] \geq 0, \ \forall v_1, v_1', v_2, \quad \text{\bf Agent 1, SP} \\
 &  [ v_2(a) \cdot f_a(v_1,v_2) + v_2(b) \cdot f_b(v_1,v_2) - p_2(v_1,v_2) ] - [ v_2(a) \cdot f_a(v_1,v_2') \\
  &  \qquad + v_2(b) \cdot f_b(v_1,v_2') - p_2(v_1,v_2') ] \geq 0, \ \forall v_1, v_2, v_2', \quad \text{\bf Agent 2, SP}
\end{flalign*}
\begin{flalign*}
 &  f_a(v) + f_b(v) = 1, \ \forall v \in V, \quad \text{\bf SCF} \\
 &  p_1(v) + p_2(v) = 0, \ \forall v \in V, \quad \text{\bf Budget Balance}
\end{flalign*}
\begin{flalign*}
 &  \ell + (v_1(a) + v_2(a)) \cdot f_a(v) + (v_1(b) + v_2(b)) \cdot f_b(v) \\
  &  \qquad \geq \max_{x \in \{a,b\}} (v_1(x) + v_2(x)), \ \forall v \in V, \quad \text{\bf Max Inefficiency} \\
  &  f_a(v), f_b(v) \geq 0, \ \forall v \in V
\end{flalign*}

\begin{proof}
 For $k=3$, each agent has $3^2 = 9$ valuations, since the number of alternatives is $2$, and therefore, the number of valuation profiles is $81$. The optimization variables are
 \begin{gather*}
  \mathbf{x} := (f_a(v^0), f_b(v^0), \ldots, f_a(v^{80}), f_b(v^{80}), p_1(v^0), p_2(v^0), \ldots, p_1(v^{80}), p_2(v^{80}), \ell )^\top.
 \end{gather*}
 Here the $81$ valuation profiles are indexed from $0$ to $80$ and are denoted by the superscripts. Hence there are $81 \times 4 + 1 = 325$ variables to the discretized relaxation of the primal problem of \Cref{eq:LP}. However, we can significantly reduce the number of variables using the symmetry of the LP. The symmetry that we consider are {\em anonymity}, i.e., the SCF alternative is invariant to the permutation of the agents, and the payments are permutated according to the permutation of the agents (defined below), and {\em neutrality}, i.e., the relabeling of the alternatives changes the alternative according to the same relabeling (\Cref{def:neutrality}).

\begin{definition}[Anonymity]
\label{def:anonymity}
 A mechanism $\langle f, \mathbf{p} \rangle$ is {\em anonymous} if for every permutation of the agents $\lambda$ we have
 $$f(\lambda(v)) = f(v) \quad \text{ and } \quad p_{\lambda(i)}(\lambda(v)) = p_i(v), \quad \forall v \in V, \forall i \in N.$$
\end{definition}
Note that, similar to \Cref{def:neutrality}, we have overloaded the notation of $\lambda$ for the valuation profile to denote that the profile where the agents are permutated according to $\lambda$.

 \begin{lemma}
  \label{lemma:symmetry}
  For every strategyproof, budget-balanced, randomized mechanism that achieves the minimum absolute inefficiency, there exists an {\em anonymous, neutral}, strategyproof, budget-balanced, randomized mechanism that achieves the same absolute inefficiency.
 \end{lemma}
 \begin{proof}
 We prove this for two agents and two alternatives. The same argument generalizes to any number of agents and alternatives. Consider an optimal solution of the optimization problem of \Cref{eq:LP}. This yields a solution $\mathbf{x}^*$ (say). Suppose, we relabel the agents $1$ and $2$ by swapping their identities, which changes the valuation profiles accordingly. For example, now the payment $p_1(v_1,v_2)$ is swapped with $p_2(v_2,v_1)$. We keep the SCF alternatives identical, i.e., $f_a(v_1,v_2) = f_a(v_2,v_1)$. Now consider the resulting vector of variables $\mathbf{x}^*_{\text{AGENT-SWAP}}$. Note that, this permutation of the variables reorders the set of constraints in \Cref{eq:LP}. The SP constraints of agent $1$ now becomes the SP constraints of agent $2$ and vice-versa. SCF constraints remain identical, budget balance constraints are reordered but same, and the max-inefficiency constraints are also reordered. Hence $\mathbf{x}^*_{\text{AGENT-SWAP}}$ is a feasible solution of the LP (\Cref{eq:LP}) and since $\mathbf{x}^*_{\text{AGENT-SWAP}}$ and $\mathbf{x}^*$ has the same value for $\ell$, $\mathbf{x}^*_{\text{AGENT-SWAP}}$ is an optimal solution of the LP (\Cref{eq:LP}).

 Similarly, we swap the alternatives $a$ and $b$ and the valuations accordingly to obtain a different reordered vector $\mathbf{x}^*_{\text{ALT-SWAP}}$. This relabeling of the alternatives again reorders all the constraints in a different way than the earlier case, with $\ell$ remaining same in both these cases. In a similar way as before, we argue that $\mathbf{x}^*_{\text{ALT-SWAP}}$ is an optimal solution of the LP (\Cref{eq:LP}).

 Now, we swap both the alternatives and agents to obtain $\mathbf{x}^*_{\text{AGENT-ALT-SWAP}}$ which reorders the constraints in a two-fold manner, but the last variable of this vector remains $\ell$ as before, and therefore it is also an optimal solution of the LP (\Cref{eq:LP}).

 Now, we have $4$ optimal solutions given the original optimal solution $\mathbf{x}^*$, which are complementary to each other in terms of agents and alternatives, but all of them are strategyproof, budget-balanced, randomized mechanisms (since they are feasible solutions of \Cref{eq:LP}). Consider the average of all these solutions: $$\mathbf{x}^{A,N} = \frac{1}{4} (\mathbf{x}^* + \mathbf{x}^*_{\text{AGENT-SWAP}} + \mathbf{x}^*_{\text{ALT-SWAP}} + \mathbf{x}^*_{\text{AGENT-ALT-SWAP}}). $$
 By construction, $\mathbf{x}^{A,N}$ is anonymous and neutral, but this is also another optimal solution of the LP (\Cref{eq:LP}) (since the set of constraints is convex). Hence, we have proved the lemma for two agents and two alternatives. For $n$ agents and $m$ alternatives, we consider all $n!$ and $m!$ possible permutations of the agents and alternatives respectively and take the mean of them to obtain our resulting optimal solution that is both anonymous and neutral.
 \end{proof}

 Hence, it is WLOG to consider neutral and anonymous mechanisms to solve the optimization problem of \Cref{eq:LP}. This reduces the number of variables in the primal problem, since for valuations that are either agent permutated or alternative permutated or both permutated version of a valuation profile we have already considered, we can replace their constraints with the already considered variables. We write the coefficient matrix of the constraint set of the LP in \Cref{eq:LP} denoted by $A$ as follows.

\begin{gather}
\label{eq:matrix}
\arraycolsep=1.4pt\def\arraystretch{1.05}
\begin{array}{cccccccccccccccc}
   & f_a(v^0) & f_b(v^0) & \ldots & \ldots & f_a(v^{80}) & f_b(v^{80}) & \
  p_1(v^0) & p_2(v^0) & \ldots & \ldots & p_1(v^{80}) & p_2(v^{80}) & \ell & \\ \hdashline[1pt/5pt]
   \ldelim({14}{0.5em} & \
  v^0_1(a) & v^0_1(b) & -v^8_1(a) & -v^8_1(b) & 0 & 0 &  \
  -1 & 0 & 1 & 0 & 0 & 0 & 0 & \rdelim){14}{0.5em} \\
   & \
  v^0_2(a) & v^0_2(b) & -v^1_2(a) & -v^1_2(b) & 0 & 0 &  \
  0 & -1 & 0 & 1 & 0 & 0 & 0 & \\
   & \
  0 & 0 & \ldots & \ldots & 0 & 0 &  \
  0 & 0 & \ldots & \ldots & 0 & 0 & 0 & \\
   & \
  0 & 0 & -v^{79}_2(a) & -v^{79}_2(b) & v^{80}_2(a) & v^{80}_2(b) &  \
  0 & 0 & 0 & 1 & 0 & -1 & 0 & \\
   & \
  0 & 0 & \ldots & \ldots & 0 & 0 &  \
  0 & 0 & \ldots & \ldots & 0 & 0 & 0 & \\ \hdashline[1pt/5pt]
   & \
  1 & 1 & 0 & 0 & 0 & 0 &  \
  0 & 0 & 0 & 0 & 0 & 0 & 0 & \\
   & \
  0 & 0 & \ldots & \ldots & 0 & 0 &  \
  0 & 0 & 0 & 0 & 0 & 0 & 0 & \\
   & \
  0 & 0 & 0 & 0 & 1 & 1 &  \
  0 & 0 & 0 & 0 & 0 & 0 & 0 & \\ \hdashline[1pt/5pt]
   & \
  0 & 0 & 0 & 0 & 0 & 0 &  \
  1 & 1 & 0 & 0 & 0 & 0 & 0 & \\
   & \
  0 & 0 & 0 & 0 & 0 & 0 &  \
  0 & 0 & \ldots & \ldots & 0 & 0 & 0 & \\
   & \
  0 & 0 & 0 & 0 & 0 & 0 &  \
  0 & 0 & 0 & 0 & 1 & 1 & 0 & \\ \hdashline[1pt/5pt]
   & \
  w^0(a) & w^0(b) & 0 & 0 & 0 & 0 &  \
  0 & 0 & 0 & 0 & 0 & 0 & 1 & \\
  & \
  0 & 0 & \ldots & \ldots & 0 & 0 &  \
  0 & 0 & 0 & 0 & 0 & 0 & 1 & \\
   & \
  0 & 0 & 0 & 0 & w^{80}(a) & w^{80}(b) &  \
  0 & 0 & 0 & 0 & 0 & 0 & 1 &
\end{array}
\end{gather}

Where $w^p(x) = v^p_1(x) + v^p_2(x), x \in \{a, b\}$, and $p$ denotes the profile index. The header of the matrix shows the primal variables. The sections showed in dotted lines corresponds to the strategyproofness, valid SCF, budget balance, and maximum inefficiency constraints respectively. The RHS of the constrained inequalities of the LP is a vector $\mathbf{b}$ that looks as follows:
\[\mathbf{b} := (0, \ldots, 0, \mathbf{1}_{|V|}, \mathbf{0}_{|V|}, \max_x w^0(x), \ldots, \max_x w^{80}(x))^\top.\]
Denoting the cost vector of the LP as, $\mathbf{c} := (\mathbf{0}_{4|V|}, 1)^\top$, we can represent the LP of \Cref{eq:LP} in the standard form:
\begin{equation}
 \label{eq:LP-primal}
 \text{\bf primal} \qquad
 \begin{aligned}
  \min_{\mathbf{x}} && \mathbf{c}^\top \mathbf{x} \\
  \text{s.t.} && A \mathbf{x} \geq \mathbf{b}
 \end{aligned}
 \hspace{2cm}
 \text{\bf dual} \qquad
 \begin{aligned}
  \max_{\mathbf{y}} && \mathbf{b}^\top \mathbf{y} \\
  \text{s.t.} &&  \mathbf{y}^\top A \leq \mathbf{c}^\top
 \end{aligned}
\end{equation}
Our goal is to provide a lower bound of the optimal value of the primal. Hence, we consider its dual, and provide a feasible solution. By weak duality lemma, the value of the dual objective at that feasible point will be a lower bound of the primal. The dual variables represented by $\mathbf{y}$ consists of $(\lambda, \gamma, \mu, \delta)$. The $\lambda$ variables refer to the dual variables corresponding to the strategyproofness constraints, and we denote the dual variable that represent the strategyproofness of agent $i$ between the profiles $v^k$ and $v^l$ by $\lambda_{i,v^k, v^l}$. By this representation, we consider only such pairs of profiles $v^k$ and $v^l$ where only agent $i$'s valuation changes. The $\gamma$ variables are the dual variables corresponding to the constraint that the SCF must add to unity, and we denote the dual variable corresponding to value profile $v$ as $\gamma_v$. Dual variables $\mu$ and $\delta$ corresponds to the budget balance and the maximum inefficiency constraints. Since SCF and budget balance constraints are equalities, $\gamma$ and $\mu$ are unrestricted, while $\lambda$ and $\delta$ are non-negative. Additionally, in the primal problem the payment variable $p_i$'s were unrestricted, hence in the dual the corresponding constraints are equalities.

We now provide a dual feasible solution, which is represented with respect to the reduced set of dual variables. Using symmetry according to \Cref{lemma:symmetry}, we reduce the number of valuation profiles. We number the profiles from $0$ to $80$ in the following way: for the valuation $(-0.5,-0.5)$ of agent $1$, all possible valuations of agent $2$ from $(-0.5,-0.5)$ to $(0.5,0.5)$ (9 profiles) are listed, and then the valuation of agent $1$ is moved to $(-0.5,0)$. Due to symmetry, setting a primal variable $f_a(v)$ to a certain value also fixes $3$ other variables that are agent-swapped or alternative-swapped or both-swapped versions of this variable. Denote the reduced set of valuation profiles by $V_R$. This also reduces the dual variables $\gamma, \mu, \delta$ from $81$ to $27$ independent variables. However, for the $\lambda$ variables we need to list all of them since they correspond to constraints that involve two valuation profiles. Consider the following set of dual variables (numbers of $v$ and $v'$ correspond to the valuation profile numbers in the listing discussed above):
\[
 \begin{array}{ccc|c}
  i & v & v' & \lambda \\ \hline
  1 & 11 & 2 & 4/7 \\
  1 & 12 & 21 & 4/7 \\
  1 & 30 & 66 & 1/14 \\
  1 & 52 & 16 & 2/7 \\
  1 & 57 & 30 & 1/7 \\
  1 & 60 & 33 & 3/14 \\
  1 & 68 & 32 & 2/7 \\
  1 & 78 & 15 & 4/7 \\
  2 & 12 & 16 & 4/7 \\
  2 & 14 & 10 & 1/14 \\
  2 & 18 & 26 & 4/7 \\
  2 & 20 & 19 & 3/14 \\
 \end{array}
 \qquad
 \begin{array}{c|c}
  v & \gamma \\ \hline
  2 & 2/7 \\
  6 & 3/28 \\
  7 & 1/7 \\
  8 & 2/7 \\
  10 & -1/28 \\
  11 & 2/7 \\
  14 & -3/14 \\
  19 & -4/7 \\
  24 & -1/7
 \end{array}
 \qquad
 \begin{array}{c|c}
  v & \mu \\ \hline
  6 & -1/14 \\
  7 & 1/7 \\
  8 & 4/7 \\
  11 & -4/7 \\
  12 & -4/7 \\
  14 & 3/14 \\
  24 & 2/7
 \end{array}
 \qquad
 \begin{array}{c|c}
  v & \delta \\ \hline
  2 & 2/7 \\
  6 & 1/7 \\
  11 & 4/7
 \end{array}
\]
All other entries of the variables that are not shown in the list above are zero. Note that for only the $\lambda$ variables, the valuation profiles listed go beyond the index $27$, but for all other dual variables they are represented by the $27$ independent variables listed in $V_R$.

We claim that this is a feasible solution of the dual. The proof requires an exhaustive verification for each of the inequalities in the constraint set of the dual. However, we provide a few cases to give an insight how this example is picked. Consider, the variables $\lambda(1,52,16) = 2/7$ and $\lambda(1,68,32) = 2/7$. Note that, $v^{24} = ((0, 0.5), (0.5, 0)) = v^{52}$ and $v^{68} = ((0.5, 0), (0, 0.5))$ is an alternative swapped version of $v^{24}$. Also, none of the other variables involve any agent or alternative or both swap of this profile in the example we gave. Therefore, now we need to concentrate on the column $f_b(v^{24})$ in the matrix of \Cref{eq:matrix}. Note that the matrix of \Cref{eq:matrix} is also reduced on the column and the rows. On the column, each of the $f$ and $p$ columns are reduced to $|V_R|$, and on the rows, only the strategyproofness constraints retain the original number, but the SCF, budget balance and maximum inefficiency constraints reduce to $|V_R|$ in size. Carrying out the product with the terms we get $0.5 \times 2/7 + 0 \times 2/7 = 1/7$. While inspecting other variables, we find $\gamma(v^{24}) = -1/7$. Hence the sum of the products on the column $f_b(v^{24})$ gives $1/7 - 1/7 = 0$ which satisfies the inequality. This is not an isolated case, in all the columns $f_x(v)$ (the numbers of such variables are reduced because of symmetry), the examples are chosen such that the sum of the non-zero products in the SP constraints section and one non-zero product in the SCF constraints section add up to a non-positive number (for example, repeat the same argument for $v^{78}$ and $v^{19}$).

Similarly, consider the column $p_2(v^{12})$: the variable $\lambda(2,12,16) = 4/7$ gets multiplied with $-1$ in this column since the constraint for agent $2$ in the profile $v^{12}$ gives a $-1$ coefficient for $p_2(v^{12})$. However, the variable $\mu(v^{12}) = -4/7$ which is multiplied with $1$ in this column, and we can inspect that no other product is non-zero on this column. Hence the sum of the products is $-8/7$ which is non-positive, and satisfies the dual constraint.

The easiest thing to verify is the last column, where the sum of the $\delta_v$ for the reduced set of $v$'s add to unity ($2/7 + 1/7 + 4/7$). Therefore, the example provided is a dual feasible solution. We compute the objective value of the solution:
\begin{align*}
 \lefteqn{\sum_{v \in V_R} \gamma_v + \sum_{v \in V_R} \delta_v \max_{x \in \{a,b\}} w^v(x)} \\
  & = \frac{2}{7} + \frac{3}{28} + \frac{1}{7} + \frac{2}{7} - \frac{1}{28} + \frac{2}{7} - \frac{3}{14} - \frac{4}{7} - \frac{1}{7} + 0.5 \times \frac{4}{7} \\
  & = \frac{1}{7}
\end{align*}
This completes the proof of the theorem.
\end{proof}


\end{document}